\documentclass[reprint,aps,pra,superscriptaddress,amsmath,amssymb,floatfix,nofootinbib]{revtex4-2}
\usepackage[T1]{fontenc}
\usepackage{microtype}
\ifdefined\pdfminorversion\pdfminorversion=7\fi
\usepackage{xcolor}
\usepackage{enumerate}

\usepackage{graphicx,bm,booktabs}
\usepackage[caption=false,font=footnotesize,justification=raggedright,singlelinecheck=false]{subfig} 
\usepackage{tikz}
\usetikzlibrary{arrows.meta,positioning}
\usepackage[colorlinks,allcolors=blue!60!black,pdfusetitle]{hyperref}
\newcommand{\ket}[1]{\lvert#1\rangle}
\newcommand{\bra}[1]{\langle#1\rvert}
\newcommand{\logicalX}{\overline{X}}
\newcommand{\logicalZ}{\overline{Z}}
\newcommand{\logicalPsi}{\overline{\Psi}}
\newcommand{\cC}{\mathcal C}
\newcommand{\cX}{\mathcal X}
\newcommand{\cE}{\mathcal E}
\newcommand{\cN}{\mathcal N}
\newcommand{\cI}{\mathcal I}
\newcommand{\HH}{\mathsf H}
\newcommand{\CZ}{\mathsf{CZ}}
\newcommand{\GHZ}[1]{\ket{\mathrm{GHZ}_{#1}}}
\newcommand{\PH}{\mathsf{PH}}
\newcommand{\BPP}{\mathsf{BPP}}
\newcommand{\BQP}{\mathsf{BQP}}
\newcommand{\PP}{\mathsf{PP}}
\newcommand{\NP}{\mathsf{NP}}
\newcommand{\post}{\mathsf{post}}
\newcommand{\poly}{\operatorname{poly}}
\newcommand{\Vrotated}[1]{\mathcal{V}^{\mathrm{rot}}_{#1}}
\newcommand{\Erotated}[1]{\mathcal{E}^{\mathrm{rot}}_{#1}}
\newcommand{\Vstandard}[1]{\mathcal{V}^{\mathrm{std}}_{#1}}
\newcommand{\Erotatedstar}[1]{\mathcal{E}^{\mathrm{rot},*}_{#1}}
\newcommand{\Caux}[1]{\mathsf{C}^{\mathrm{std},(#1)}_{D}}
\newcommand{\Crotated}[1]{\mathsf{C}^{\mathrm{rot}}_{#1}}
\newcommand{\Cstandard}[1]{\mathsf{C}^{\mathrm{std}}_{#1}}
\newcommand{\Sample}{\mathsf{Sample}}
\newtheorem{theorem}{Theorem}
\newtheorem{lemma}{Lemma}
\newtheorem{appendixlemma}{Lemma}[section]
\tikzset{site/.style={circle,draw,fill=white,minimum size=6mm,inner sep=1pt},
 check/.style={rectangle,draw,fill=gray!12,minimum size=5mm,inner sep=2pt}}
\begin{document}
\title{Classical simulation of coherent crosstalk in surface codes}
\author{Andrew S. Darmawan}
\affiliation{Independent Researcher, Kyoto, Japan}
\author{Yelyzaveta Kolesnyk}
\affiliation{Quantum Research Center, Technology Innovation Institute (TII), Abu Dhabi, United Arab Emirates}
\affiliation{Department of Mathematics, Technical University of Munich, 85748 Garching, Germany}
\affiliation{Munich Center for Quantum Science and Technology (MCQST), 80799 M\"unchen, Germany}
\author{Robert K\"onig}
\affiliation{Department of Mathematics, School of Computation, Information and Technology, Technical University of Munich, 85748 Garching, Germany}
\affiliation{Munich Center for Quantum Science and Technology (MCQST), 80799 M\"unchen, Germany}
\date{September 30, 2026}
\begin{abstract}
We give a polynomial-time classical algorithm which samples from the syndrome distribution of a surface code state corrupted by coherent nearest-neighbor $ZZ$~crosstalk. Our algorithm complements the known efficient simulation algorithms for coherent single-qubit errors in the surface code. When both single-qubit coherent noise and coherent crosstalk are present, we find a complexity-theoretic obstruction to efficient simulation: unless the polynomial hierarchy collapses, there is no efficient classical algorithm for sampling from the syndrome distribution, even up to a constant multiplicative error. This is obtained by connecting the syndrome distribution under coherent noise to the output distribution of certain IQP circuits associated with a non-planar graph.
\end{abstract}
\maketitle

\section{Introduction}
\label{sec:intro}

Surface codes~\cite{Kitaev1997,BravyiKitaev98,Dennis2002} are among the most promising proposals for a robust quantum memory. Superconducting-qubit experiments have realized surface codes with distance up to seven, with  logical lifetimes exceeding those of the best constituent physical qubit~\cite{Krinner22,Zhao22,GoogleQAI23,GoogleQAI25}. Yet basic questions about their robustness against experimental noise remain.

One such question concerns coherent crosstalk. Residual $ZZ$ interactions -- referred to as coherent crosstalk -- between neighboring superconducting qubits are a major source of coherent errors~\cite{Sheldon16,Mundada19,Krinner20,Sarovar20,Kandala21,ZhaoEtAl22,Tripathi22}, including in recent surface-code experiments~\cite{Vezvaee26}. Such a coupling acts coherently by the unitary~$e^{i\theta Z_qZ_{q'}}$. Compared to  probabilistic Pauli noise,
 interference arising from coherent noise can substantially change the resulting logical error rate~\cite{GutierrezEtAl2016,GreenbaumDutton2018,BravyiEnglbrechtKoenigPeard18,VennBeri20}. We ask whether the surface code, with the  efficient minimum-weight perfect matching (MWPM) decoder, can withstand this noise. We consider the code-capacity setting: coherent crosstalk acts on a prepared code state, followed by noise-free syndrome extraction and recovery. Noise-resilience is quantified by the resulting logical channel's diamond-norm distance from the identity, averaged over syndromes; we are interested in  its scaling with the number of physical qubits as well as the noise strength expressed by the magnitude~$|\theta|$.

{\em Prior work.} Surface-code performance is best understood under stochastic
Pauli noise. The Gottesman--Knill theorem~\cite{Gottesman98,AaronsonGottesman04}
permits exact polynomial-time stabilizer simulation~\cite{Gidney21}, enabling
numerical threshold estimates~\cite{Dennis2002,WangHarringtonPreskill03,
RaussendorfHarrington07,Fowler2012}. Analytic threshold estimates are also
available for stochastic Pauli noise through mappings to disordered
statistical-mechanics models~\cite{Dennis2002}.
Local stochastic Pauli noise is therefore the standard first test of a
stabilizer code's viability.

Coherent noise is more difficult to simulate because interference enters the syndrome probabilities. Exact tensor-network contraction of coherent $ZZ$ crosstalk during syndrome extraction has been limited to distance~$d=3$ ($17$ qubits)~\cite{HuangEtAl20}. Methods that reach larger distances generally rely on approximations without rigorous error bounds. These include truncated tensor networks for general local noise~\cite{DarmawanPoulin17}, a hybrid stabilizer--matrix-product-state method applied to coherent crosstalk up to distance~$d=9$~\cite{HarperEtAl26}, and perturbative effective error models reaching distance~$d=7$~\cite{HinesEtAl26}. Related matrix-product-state studies have considered single-qubit rotations about a generic axis~\cite{BehrendsBeri25} and combined single- and two-qubit coherent noise~\cite{BaoAnand24}, again without rigorous a priori error bounds. Quasiprobability sampling gives unbiased estimates for coherent dephasing up to distance~$d=11$, but its cost grows exponentially with the non-Clifford content of the noise~\cite{LeBlondEtAl25}. Pauli twirling makes stabilizer simulation efficient, but changes the noise model by discarding interference~\cite{ZhouJiDing25}.

An important case in which efficient classical simulation is possible without approximation is coherent single-qubit $Z$ noise. A mapping to fermionic linear optics gives a polynomial-time algorithm that samples the syndrome distribution exactly for rotated surface codes~\cite{BravyiEnglbrechtKoenigPeard18}. Subsequent extensions apply to general planar graphs, including standard surface codes with smooth and rough boundaries~\cite{VennBeri20}. These algorithms have enabled numerical threshold studies on codes with more than $2\,000$ physical qubits~\cite{BravyiEnglbrechtKoenigPeard18,VennBeri20,VennBehrendsBeri23}.

Thus methods general enough to include coherent crosstalk were restricted to small codes, relied on empirically controlled approximations, incurred exponential sampling cost, or replaced the coherent model by its Pauli twirl. Efficient exact algorithms were known for single-qubit coherent noise, but not for nearest-neighbor coherent crosstalk on rotated surface codes. To our knowledge, no efficient exact simulation algorithm was available for nearest-neighbor coherent crosstalk at distances large enough to study threshold scaling. Our algorithm enables such a scaling analysis and model-dependent threshold extrapolations, whose limitations are discussed in Sec.~\ref{sec:numerics}.

\section{Our contribution}

{\em An efficient simulation algorithm for coherent crosstalk.}
We give a polynomial-time classical algorithm which takes as input the coupling strengths~$\{\theta_e\}_{e\in\cE}$ of nearest-neighbor $ZZ$-interactions (generating coherent crosstalk) on a rotated surface code of odd distance~$d$, and achieves the following: It (i)~samples from the distribution of syndromes obtained when the unitary $U_{\mathrm{nn}}(\bm\theta)=\prod_{e=\{q,q'\}\in\cE}e^{i\theta_eZ_qZ_{q'}}$ is applied to a code state and all stabilizers are measured, and (ii)~computes the final logical operation once an (efficiently computable) syndrome-dependent Pauli correction is applied to map back to the code space. The coupling strengths~$\{\theta_e\}_{e\in\cE}\subset\mathbb{R}$ may be arbitrary and may vary from edge to edge, where an edge~$e\in\cE$ of the surface code lattice connects two neighboring qubits. The cost per syndrome sample is~$O(d^6)$ elementary arithmetic operations. In our implementation, drawing one syndrome sample at $d=37$ ($1\,369$ physical qubits) takes about $5$~ms (see Table~\ref{tab:runtime}).

Using Algorithm~1, we assess the accuracy of the minimum-weight perfect matching decoder under coherent crosstalk in rotated surface codes at odd distances $d\in \{13,21,29,37\}$. We simulate error recovery with uniform coupling $\theta_e=\theta$ (for all $e\in \cE$) and compare with the incoherent model in which each interaction is replaced by its Pauli twirl, i.e., by a $Z_qZ_{q'}$ error occurring independently with probability $\sin^2\theta$ on each edge $e=\{q,q'\}\in\cE$ (see Sec.~\ref{sec:numerics}). We find that the residual logical error under coherent crosstalk is substantially larger than under the per-interaction Pauli twirl at the same angle. A joint finite-size scaling fit gives threshold extrapolations of $\theta_c/\pi=0.0342(4)$ under coherent crosstalk and $\theta_c/\pi=0.0493(5)$ under the per-interaction Pauli twirl. The corresponding values of $\sin^2\theta_c$ are approximately $1.1\%$ and $2.4\%$, respectively. These extrapolations depend on the scaling ansatz and the specified decoder: the quoted spreads measure sensitivity to the included distances, while the poor goodness of fit prevents controlled asymptotic threshold estimates (Sec.~\ref{sec:numerics}).

{\em Hardness of sampling under combined coherent noise.}
Our second result concerns the combination of single-qubit $Z$-rotations and nearest-neighbor $ZZ$ crosstalk, i.e., noise described by the unitary
\begin{align}
 U_{\mathrm{comb}}(\bm\varphi,\bm\theta)
 =\left(\prod_q e^{i\varphi_q Z_q}\right)
  \left(\prod_{e=\{q,q'\}\in\cE}e^{i\theta_eZ_qZ_{q'}}\right)\ .
 \label{eq:introcombined}
\end{align}
For a distance-$d$ rotated code $\Crotated{d}$ prepared in the state $\ket{\overline{0}}$, write $p(s)=p_{\ket{\overline{0}}}^{\Crotated{d}}(s\mid U_{\mathrm{comb}}(\bm\varphi,\bm\theta))$ for the probability of observing syndrome $s$ under ideal measurement. The superscript specifies the code family and distance, and the argument the noise. The input subscript is retained because combined noise can give a syndrome distribution that depends on the encoded state (Sec.~\ref{sec:hardness}). We show that the existence of a polynomial-time classical algorithm producing
-- for every allowed choice $(\bm\varphi,\bm\theta)$ of angles --
a sample from a distribution $\widetilde p$ satisfying, for a fixed $0\leq\epsilon<1$,
\begin{align}
 \bigl|\widetilde p(s)
       -p(s)\bigr|
 \leq\epsilon\,p(s)
 \label{eq:introapprox}
\end{align}
for every syndrome~$s$ would imply a collapse of the polynomial hierarchy to its third level (Theorem~\ref{thm:hardness} in Sec.~\ref{sec:hardness}).

{\em Three computational regimes.}
Single-qubit coherent $Z$-rotations and nearest-neighbor coherent $ZZ$-crosstalk each admit efficient exact simulation, whereas the syndrome distribution resulting from their combination is hard to sample from  under the multiplicative approximation guarantee above (see Eq.~\eqref{eq:introapprox}), unless the polynomial hierarchy collapses. Thus the noise model affects not only the logical error rate but also the computational complexity of assessing code performance. This distinction complements the numerical evidence of Ref.~\cite{BaoAnand24} for an entanglement barrier to matrix-product-state decoding under combined single-qubit and two-qubit rotations: our hardness result concerns syndrome sampling itself.

\subsection{Techniques}
{\em Reducing crosstalk to two independent single-qubit noise problems.}
Both single-qubit $Z$-rotations and nearest-neighbor $ZZ$ rotations are generated by commuting terms, each of which flips at most two $X$-stabilizers. Our algorithm exploits additional structure in the syndromes of the nearest-neighbor $ZZ$-errors. On a rotated code, the $X$-stabilizers split into two disjoint sets $A$ and $B$ such that every nearest-neighbor $ZZ$-error has its syndrome entirely within one set. Each set is in one-to-one correspondence with the $X$-stabilizers of a standard (unrotated) surface code of distance $(d+1)/2$. Under this correspondence, each physical $ZZ$-rotation has the same syndrome action as a single-qubit $Z$-rotation on one auxiliary code; interactions with the same syndrome combine by adding their angles.

As a consequence, the resulting syndrome distribution factorizes exactly as
\begin{align}
 p^{\Crotated{d}}(s\mid U_{\mathrm{nn}})=p^{\Cstandard{D}}(s^{(A)}\mid U^{(A)})p^{\Cstandard{D}}(s^{(B)}\mid U^{(B)})\ ,
 \label{eq:introfactor}
\end{align}
where $D=(d+1)/2$ and $s^{(A)}$, $s^{(B)}$ are the syndromes of the two auxiliary standard (unrotated) codes under their assigned single-qubit noise $U^{(A)}$, $U^{(B)}$ (Sec.~\ref{sec:mapping}, Eq.~\eqref{eq:factor}). No encoded input appears in Eq.~\eqref{eq:introfactor}: both sides are the same for any encoded state, on the original as well as on the auxiliary codes. One can therefore draw $s^{(A)}$ and $s^{(B)}$ independently and combine them to obtain an exact syndrome sample of the original code. We use the single-qubit simulation algorithms of Refs.~\cite{BravyiEnglbrechtKoenigPeard18,VennBeri20} as subroutines; the planar-graph construction of Ref.~\cite{VennBeri20} applies to the auxiliary standard (unrotated) surface codes.

The reduction of crosstalk simulation to single-qubit noise described here for rotated codes can similarly be carried out for odd-distance unrotated codes, using two auxiliary rotated codes and the polynomial-time simulator of Ref.~\cite{BravyiEnglbrechtKoenigPeard18}.

{\em Postselected universality with two-body interactions on a king lattice.}
Adding physical single-qubit rotations connects the two syndrome sets. On the syndrome register, the combined noise is equivalent to an instantaneous quantum polynomial-time (IQP) circuit: commuting diagonal gates act on $\ket+$ inputs, followed by measurements in the Hadamard basis. Together, physical single-qubit rotations and $ZZ$-rotations give an interaction graph isomorphic to a finite subgraph of a king lattice (Sec.~\ref{sec:hardness}). Thus a combination of the two physical noise types on the surface code corresponds to an IQP circuit with only two-body interactions on this lattice.

To establish hardness, we start from a universal nearest-neighbor brickwork circuit and replace its Hadamards by postselected gadgets~\cite{BremnerJozsaShepherd2011}. Expanding controlled-$Z$ gates leaves commuting $Z$ and $ZZ$ rotations on its planar space--time graph. A distributed Greenberger--Horne--Zeilinger (GHZ) reference state replaces the $Z$-rotations by $ZZ$ couplings: its two computational-basis branches give identical data probabilities in the Hadamard basis, so discarding the reference outcomes preserves the target distribution. Postselected parity measurements prepare this reference and mediate interactions along lattice paths. A fixed crossing construction routes the reference and data paths through distinct king-lattice sites, with polynomial overhead, degree at most four, and angles in multiples of $\pi/8$ (Appendix~\ref{app:hardness}). The resulting family is universal under postselection. A classical sampler with the pointwise guarantee in Eq.~\eqref{eq:introapprox} would then simulate postselected quantum computation, implying the stated collapse~\cite{Aaronson2005,BremnerJozsaShepherd2011}.

{\em Outline.}
The remainder of the paper is structured as follows. Section~\ref{sec:setup} defines the code and noise model, formulates the simulation tasks, and establishes exact Pauli recovery for crosstalk noise on odd-distance rotated codes. Section~\ref{sec:mapping} presents the reduction to single-qubit noise and the sampling algorithm, Sec.~\ref{sec:numerics} reports the numerical results, Sec.~\ref{sec:hardness} states the hardness of sampling under combined single-qubit and crosstalk noise, and Sec.~\ref{sec:discussion} discusses the scope of our results. 

\section{Problem statement}
\label{sec:setup}

\subsection{Surface codes}
\label{sec:model}

\begin{figure}[t]
\centering
\begin{minipage}[t]{118pt}
\centering
\subfloat[{Rotated surface code, \mbox{$d=5$}.}]{\begin{minipage}{\linewidth}
 \centering
 \includegraphics{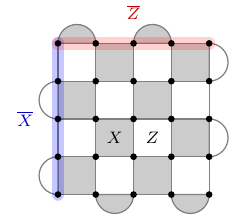}
\end{minipage}}
\end{minipage}\hfill
\begin{minipage}[t]{118pt}
\centering
\subfloat[{Standard (unrotated) surface code, $d=5$.}]{\begin{minipage}{\linewidth}
 \centering
 \includegraphics{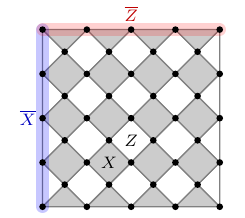}
\end{minipage}}
\end{minipage}
\caption{Surface codes of distance $d=5$. Circles are physical qubits and faces are stabilizers, grey for $X$ and white for $Z$. The blue and red strings are the logical representatives $\overline{X}$ and $\overline{Z}$. Full faces have weight four. In the rotated surface code (a), the boundary semicircles correspond to weight-two stabilizers. In the standard (unrotated) surface code (b), the boundary triangles are weight-three stabilizers. The left and right (top and bottom) sides of the standard surface code are also called rough (smooth) boundaries.}
\label{fig:surface_codes}
\end{figure}

We consider a rotated square surface code of odd distance $d\geq5$, encoding one qubit into $n=d^2$ physical qubits on a $d\times d$ array~\cite{Wen2003,BombinMartinDelgado2007,HorsmanFowlerDevittVanMeter12}. We denote this code by $\Crotated{d}$; the superscript specifies the code family and the subscript its distance (see Figure~\ref{fig:surface_codes}).

For the rotated surface code $\Crotated{d}$, label the physical qubits $q_{i,j}$, with $0\leq i,j<d$, increasing rightward in $i$ and upward in $j$. Set $D=(d+1)/2$ and define
\begin{align}
 G=\{(u,v):0\leq u\leq D-2,\ 0\leq v\leq D-1\}\ .
 \label{eq:componentgrid}
\end{align}
We index the $X$-type stabilizer generators by triples $(C,u,v)$, where $C\in\{A,B\}$ and $(u,v)\in G$. Thus, their index set is $\{A,B\}\times G$. Writing $X_{i,j}$ for a single-qubit Pauli $X$ on qubit $q_{i,j}$, the bulk stabilizer generators are
\begin{align}
 \begin{aligned}
 S^{X,(A)}_{u,v}&=\prod_{\alpha,\beta=0}^{1}X_{2u+\alpha,2v+\beta}\ ,
 &&0\leq v\leq D-2\ ,\\
 S^{X,(B)}_{u,v}&=\prod_{\alpha,\beta=0}^{1}X_{2u+1+\alpha,2v-1+\beta}\ ,
 &&1\leq v\leq D-1\ ,
 \end{aligned}
 \label{eq:bulkxchecks}
\end{align}
and the generators on the top and bottom boundaries are
\begin{align}
 \begin{aligned}
 S^{X,(A)}_{u,D-1}&=X_{2u,d-1}X_{2u+1,d-1}\ ,\\
 S^{X,(B)}_{u,0}&=X_{2u+1,0}X_{2u+2,0}\ ,
 \end{aligned}
 \label{eq:boundaryxchecks}
\end{align}
with $0\leq u\leq D-2$. The superscript label $C\in\{A,B\}$, which we refer to as a component, and $(u,v)\in G$ labels a generator within it.

Writing $Z_{i,j}$ for Pauli $Z$ on $q_{i,j}$, the $Z$-type stabilizer generators are the bulk operators
\begin{align}
 S^Z_{i,j}=\prod_{\alpha,\beta=0}^{1}Z_{i+\alpha,j+\beta}\ ,
 \quad 0\leq i,j\leq d-2,\quad i+j\text{ odd}\ ,
\end{align}
and the weight-two operators $Z_{0,2v}Z_{0,2v+1}$ and $Z_{d-1,2v+1}Z_{d-1,2v+2}$, for $0\leq v\leq D-2$, on the left and right boundaries. Encoded states have eigenvalue $+1$ for each of these commuting generators. There are $r_X=r_Z=(n-1)/2$ independent generators of each type; when coordinates are not needed, we enumerate them as $\{S_u^X\}_{u=1}^{r_X}$ and $\{S_v^Z\}_{v=1}^{r_Z}$. For brevity, we call the $X$-type and $Z$-type stabilizer generators $X$-stabilizers and $Z$-stabilizers throughout the paper.

Logical Pauli operators $\overline{X}$ and $\overline{Z}$ commute with all stabilizers and anticommute with each other. Multiplying them by stabilizers gives equivalent representatives. The distance $d$ is the minimum weight of a non-trivial logical Pauli operator, where weight counts the qubits an operator acts on non-trivially. We denote the logical-$Z$ eigenstates by $\ket{\overline{0}}$ and $\ket{\overline{1}}$, with eigenvalues $+1$ and $-1$, respectively.

For comparison, the distance-$d$ standard (unrotated) surface code $\Cstandard{d}$ has $d^2+(d-1)^2$ physical qubits and weight-three boundary stabilizers~\cite{Kitaev1997,BravyiKitaev98}. We use the representation in Fig.~\ref{fig:surface_codes}(b), with physical qubits at vertices and grey and white faces for $X$- and $Z$-stabilizers. Figure~\ref{fig:surface_codes} shows both families at distance five.

For either family, $n(\cC)$ denotes the number of physical qubits and $r_X(\cC)$ the number of independent $X$-stabilizers:
\begin{align}
 \begin{aligned}
 n(\Crotated{d})&=d^2 & n(\Cstandard{d})&=d^2+(d-1)^2,\\
 r_X(\Crotated{d})&=(d^2-1)/2 & r_X(\Cstandard{d})&=d(d-1).
 \end{aligned}
 \label{eq:codesizes}
\end{align}
Both codes encode one logical qubit, so their code spaces have dimension two.

\subsection{Problem formulation: coherent noise and recovery}
\label{sec:problem}

Label the $n$ physical qubits by $[n]=\{1,\ldots,n\}$. For any subset $\Gamma\subseteq[n]$, define the $n$-qubit Pauli operator $Z(\Gamma)=\prod_{q\in\Gamma} Z_q$. We call such an operator a $Z$-type Pauli operator, and refer to $\mathsf{supp}(Z(\Gamma))=\Gamma$ as  the support of $Z(\Gamma)$.
We consider coherent single-qubit (sq) rotations, nearest-neighbor (nn) $ZZ$ crosstalk and their combination (comb), i.e.,
\begin{align}
 U_{\mathrm{sq}}(\bm\varphi)&=\prod_q e^{i\varphi_q Z_q}\ ,\nonumber\\
 U_{\mathrm{nn}}(\bm\theta)&=\prod_{e=\{q,q'\}\in\cE}e^{i\theta_e Z_qZ_{q'}}\ ,\nonumber\\
 U_{\mathrm{comb}}(\bm\varphi,\bm\theta)&=U_{\mathrm{sq}}(\bm\varphi)U_{\mathrm{nn}}(\bm\theta)\ .
 \label{eq:noise}
\end{align}
Here $\bm\varphi=(\varphi_q)_{q\in[n]}$ and $\bm\theta=(\theta_e)_{e\in\cE}$ are collections of real angles that may vary independently. The set $\cE$ consists of the edges of the lattice in Fig.~\ref{fig:surface_codes}(a); each edge $e=\{q,q'\}\in\cE$ connects two horizontally or vertically neighboring physical qubits $q,q'$.

We assume ideal state preparation and stabilizer measurements, together with a specified correction rule. The error-correction procedure we are interested in consists of four steps:
\begin{enumerate}
 \item Prepare an encoded state $\ket{\logicalPsi}$.
 \item Apply the coherent noise $U$.

 \item Measure all stabilizers. Since $U$ is diagonal in the computational basis, it commutes with every $Z$-stabilizer, whose outcome therefore remains $+1$. The non-trivial measurement record is thus the binary vector
 \[
 s=(s_1,\ldots,s_{r_X})\in\{0,1\}^{r_X}\ ,
 \]
 of $X$-stabilizer outcomes, where $s_u=0$ denotes outcome $+1$ and $s_u=1$ denotes outcome $-1$. The corresponding projector is
 \begin{align}
 \Pi_s=\prod_{u=1}^{r_X}\frac{I+(-1)^{s_u}S_u^X}{2}\ .
 \label{eq:projector}
 \end{align}
 The syndrome $s$ occurs with probability
 \begin{align}
 p_{\Psi}^{\cC}(s\mid U)=\|\Pi_sU\ket{\logicalPsi}\|^2\ .
 \label{eq:pssample}
 \end{align}

 \item Apply a $Z$-type Pauli correction~$C_s$
 chosen to have the syndrome $s$, see  Eq.~\eqref{eq:correction} below.
\end{enumerate}

We now make the syndrome notation used in the third and fourth steps explicit. All binary-vector operations below are modulo two. For a $Z$-type Pauli error $Z(\Gamma)$, define its syndrome
\[
\partial\Gamma\in\{0,1\}^{r_X}
\]
through the commutation relations
\begin{align}
 S_u^X Z(\Gamma)
 =
 (-1)^{(\partial\Gamma)_u}
 Z(\Gamma)S_u^X\ 
 \quad\text{for all }u\in [r_X]\ .
 \label{eq:boundary}
\end{align}
Thus $(\partial\Gamma)_u=1$ exactly when the support of $S_u^X$ intersects $\Gamma$ in an odd number of physical qubits. Equivalently, $\partial\Gamma$ records the $X$-stabilizers violated after application of $Z(\Gamma)$.

A $Z$-type correction operator $C_s$ is compatible with syndrome $s$ precisely when $\partial(\mathsf{supp}(C_s))=s$, or, equivalently,
\begin{align}
 S_u^X C_s
 =
 (-1)^{s_u}C_sS_u^X
 \quad\text{for all }u\in [r_X]\ .
 \label{eq:correction}
\end{align}
Applying such a correction returns the measured state to the code space. For any syndrome $s$ with $p_{\Psi}^{\cC}(s\mid U)>0$, the resulting normalized state after correction is
\begin{align}
 \ket{\psi_s}
 =
 \frac{C_s\Pi_sU\ket{\logicalPsi}}
 {\sqrt{p_{\Psi}^{\cC}(s\mid U)}}\  .
 \label{eq:finalstate}
\end{align}
The problem of simulating noise and subsequent  recovery then boils down to the following two computational tasks:
\begin{enumerate}[(A)]
\item\label{it:computationaltaskA}
First, we must sample a syndrome $s$ according to the distribution $p_{\Psi}^{\cC}(s\mid U)$.
\item\label{it:computationaltaskB}
Second, given a sampled syndrome $s\in \{0,1\}^{r_X}$ and a compatible correction $C_s$, we must determine the resulting logical state~$\ket{\psi_s}$.
\end{enumerate} In the next section, we characterize this conditional logical action for coherent crosstalk noise.

\subsection{Logical action and recovery for coherent crosstalk}
\label{sec:recovery}

Both tasks~\eqref{it:computationaltaskA}, \eqref{it:computationaltaskB} simplify when the noise consists of coherent crosstalk errors alone, without single-qubit rotations. Whenever the syndrome distribution is independent of the encoded initial state $\ket{\overline{\Psi}}$, we omit the associated subscript and write $p^{\cC}(s\mid U)$.

\begin{lemma}[Logical action and exact recovery]
\label{lem:logicalaction}
Let $U=U_{\mathrm{nn}}$ be coherent $ZZ$-crosstalk noise (see Eq.~\eqref{eq:noise}) on  a rotated surface code $\Crotated{d}$ of odd distance~$d$. Let $C_s$ be any correction operation of $Z$-type compatible with the syndrome $s$ (see Eq.~\eqref{eq:correction}). Define
\begin{align}
 t_s=|\mathsf{supp}(C_s)|\bmod2\ .
 \label{eq:corrparity}
\end{align} 
Then there is a scalar $\kappa_s\in\mathbb{C}$ such that
\begin{align}
 C_s\Pi_sU_{\mathrm{nn}}\ket{\logicalPsi}=\kappa_s\,\logicalZ^{\,t_s}\ket{\logicalPsi}\ ,
 \qquad
 p^{\Crotated{d}}(s\mid U_{\mathrm{nn}})=|\kappa_s|^2\ ,
 \label{eq:parity}
\end{align}
for every encoded state $\ket{\logicalPsi}\in\cC$. In particular, the syndrome distribution $p^{\Crotated{d}}(s\mid U_{\mathrm{nn}})=|\kappa_s|^2$ does not depend on the encoded state. The conditional logical operation defined by Eq.~\eqref{eq:finalstate} 
(mapping the original state~$\ket{\logicalPsi}$ to the state $\ket{{\psi}_s}$) is -- up to an irrelevant global phase --  the identity when $t_s=0$ and $\logicalZ$ when $t_s=1$, that is,
\begin{align}
 \ket{{\psi}_s}\propto \logicalZ^{\,t_s}\ket{\logicalPsi}\ .
 \label{eq:rotatedlogicalrotation}
\end{align}
\end{lemma}
We note that because we are working with $\Crotated{d}$ where $d$ is odd, the correction~$C_s$ can always be chosen such that $t_s=0$: for example, if $C_s$ has odd weight, replace it with
$ C_s^{\mathrm{even}}=\logicalZ^{\,t_s}C_s$.
This gives $t_s=0$ for every syndrome~$s$, hence the recovery  is exact for any 
 rotation angles for this choice.  This consequence of Lemma~\ref{lem:logicalaction} is in line with the observation of Ref.~\cite{WangEtAl25} that stochastic nearest-neighbor $ZZ$ errors on rotated surface codes are perfectly correctable.  We emphasize that this does not make error correction trivial in practice, however. The even-weight rule relies on the prior knowledge that the noise consists of $ZZ$-errors alone, and it produces a logical error on every odd-weight error, for example a single-qubit $Z$-error. Without such prior knowledge one typically uses a generic decoder such as MWPM, which performs well both under local stochastic Pauli noise~\cite{Dennis2002,FowlerThreshold12} and under single-qubit coherent noise~\cite{BravyiEnglbrechtKoenigPeard18}. Its performance under coherent crosstalk is studied in Sec.~\ref{sec:numerics}.

\emph{Proof sketch.} Every $Z$-stabilizer has even weight, whereas $\overline{Z}$ has odd weight. The noise contains only even-weight $Z$-errors; after a correction of parity $t_s$, all contributions with syndrome $s$ therefore have the same logical action $\logicalZ^{t_s}$. The detailed proof is given in Appendix~\ref{app:parityproof} $\square$.

Lemma~\ref{lem:logicalaction} means that the simulation task of determining the effective logical channel after error recovery reduces to sampling from the distribution $p(s)=p^{\Crotated{d}}(s\mid U_{\mathrm{nn}})$ of syndromes, independent of the encoded state~$\ket{\logicalPsi}$.
Once a syndrome~$s$ is drawn from~$p$, the residual logical sample can simply be read off from the decoder's output by determining the bit (parity) $t_s$.  We note that ordinary MWPM decoding does not return an even-weight correction for every syndrome~$s$; hence it produces a logical $Z$-error precisely on the sampled syndromes~$s$ with odd-weight correction~$C_s$.

The simplifications expressed by Lemma~\ref{lem:logicalaction} 
  do not hold for general coherent noise. In particular, 
  if the noise involves (Hamiltonian) $Z$-type Paulis which do not commute with~$\logicalX$, then  both the syndrome distribution and the logical state~$\ket{\psi_s}$ after correction can depend on the initial code state~$\ket{\logicalPsi}$. This already happens for single-qubit coherent rotations on standard surface codes~\cite{VennBeri20}; for odd-distance rotated surface codes, the syndrome distribution under single-qubit coherent noise is independent of the initial code state, while the conditional logical operation is generally a non-trivial $Z$-rotation~\cite{BravyiEnglbrechtKoenigPeard18}.

\subsection{Syndrome distribution and IQP circuits\label{sec:syndromeiqp}}
Here we show that the distribution of syndromes with coherent noise acting on a logical basis state~$\ket{\overline{\lambda}}$ can be rewritten as the distribution of outputs of an instantaneous quantum polynomial-time (IQP) circuit~\cite{BremnerJozsaShepherd2011}. For an arbitrary encoded initial state $\ket{\logicalPsi}$, the syndrome distribution is a convex combination of these distributions. Furthermore, the IQP circuit is generated by an Ising model on a specific graph determined by the code and noise considered.
This  observation is our main technical tool: it is used in the proof of Lemma~\ref{lem:mapping} below (which underlies our algorithm) and in Sec.~\ref{sec:hardness}, where we establish hardness of sampling in the case of general coherent noise.

The representation in terms of an IQP circuit holds for general CSS codes, without assumptions on their geometry or distance. The IQP circuits obtained below are generated by Ising-type Hamiltonians; the relation between such circuits and Ising partition functions was previously studied by Fujii and Morimae~\cite{FujiiMorimae2017}.
\begin{lemma}[Syndrome probabilities and IQP]
\label{lem:syndromeiqp}
Let $\cC$ be an $n$-qubit CSS code encoding $k$ logical qubits, with independent $X$-type stabilizer generators $S_1^X,\ldots,S_{r_X}^X$.
Let $\ket{\overline{\lambda}}$, $\lambda\in\{0,1\}^k$, denote the code state stabilized by all $X$- and $Z$-type generators and with eigenvalue $(-1)^{\lambda_j}$ for the $j$-th logical operator $\logicalZ_j$. Consider the coherent noise unitary
\[
 U=\prod_{\Gamma\in\cI}e^{i\theta_\Gamma Z(\Gamma)}\ ,
\]
where $\cI$ is any finite collection of supports $\Gamma\subseteq[n]$ of $Z$-type products on the $n$ physical qubits and $\theta_\Gamma\in\mathbb R$. Define $\partial\Gamma\in\{0,1\}^{r_X}$ as the ($X$-type) syndrome induced by the error~$Z(\Gamma)$, i.e., 
\[
 S_u^XZ(\Gamma)=(-1)^{(\partial\Gamma)_u}Z(\Gamma)S_u^X\ \ 
 \quad\text{for all }u\in [r_X]\ ,
\]
and define $\ell(\Gamma)\in\{0,1\}^k$ by
\[
 \logicalX_j Z(\Gamma)=(-1)^{\ell_j(\Gamma)}Z(\Gamma)\logicalX_j\ 
 \quad\text{for all }j\in [k]\ .
\]
Let $\Pi_s$ project onto the measured $X$-stabilizer syndrome $s$. Then,
\begin{align}
 \begin{aligned}
 p_\lambda^{\cC}(s\mid U)
 &=\|\Pi_sU\ket{\overline{\lambda}}\|^2\\
 &=\left|\bra{s}\HH^{\otimes r_X}W_{\ket{\overline{\lambda}}}\ket+^{\otimes r_X}\right|^2\\
 &\quad\text{for every }s\in\{0,1\}^{r_X}\ ,
 \end{aligned}
 \label{eq:syndromeiqpprobability}
\end{align}
where $\HH$ is the Hadamard gate and $W_{\ket{\overline{\lambda}}}$ is the unitary
\begin{align}
 W_{\ket{\overline{\lambda}}}=\prod_{\Gamma\in\cI}e^{i\theta_\Gamma(-1)^{\lambda\cdot\ell(\Gamma)}Z(\partial\Gamma)}\ .
 \label{eq:syndromeiqpcircuit}
\end{align}
Here $Z(\partial\Gamma)=\prod_{u=1}^{r_X}Z_u^{(\partial\Gamma)_u}$ acts on the syndrome register ($(\mathbb{C}^2)^{\otimes r_X}$) and $\lambda\cdot\ell(\Gamma)$ is the binary inner product. For a normalized encoded state $\ket{\logicalPsi}=\sum_\lambda a_\lambda\ket{\overline\lambda}\in\cC$, we have
\begin{align}
    \begin{aligned}
    p_{\ket{\logicalPsi}}^{\cC}(s\mid U)=\sum_\lambda |a_\lambda|^2p_\lambda^{\cC}(s\mid U)\\
    \quad\text{for every }s\in\{0,1\}^{r_X}\ .
  \end{aligned}
 \label{eq:generalinputmixture}
\end{align}
In particular, for one logical qubit (i.e. $k=1$) and $\ket{\logicalPsi}=a\ket{\overline0}+b\ket{\overline1}$,
\begin{align}
  p_{\ket{\logicalPsi}}^{\cC}(s\mid U)=|a|^2p_{\ket{\overline{0}}}^{\cC}(s\mid U)+|b|^2p_{\ket{\overline{1}}}^{\cC}(s\mid U).
  \label{eq:inputmixture}
\end{align}
\end{lemma}
Below we write $W=W_{\ket{\overline{0}}}$. By Eq.~\eqref{eq:inputmixture}, input independence for one logical qubit is equivalent to $p_{\ket{\overline{0}}}^{\cC}=p_{\ket{\overline{1}}}^{\cC}$, in which case we write $p^{\cC}$ for the common distribution. On an odd-distance rotated code, choose $\logicalX=X^{\otimes n}$. For nearest-neighbor crosstalk alone, every active support has $\ell(\Gamma)=0$, so $W_{\ket{\overline{1}}}=W_{\ket{\overline{0}}}$, recovering the input independence of Lemma~\ref{lem:logicalaction}. For single-qubit noise alone, every active support has $\ell(\Gamma)=1$, so $W_{\ket{\overline{1}}}=W_{\ket{\overline{0}}}^\dagger$. Since $W_{\ket{\overline{0}}}$ is diagonal and the input and measurement vectors in Eq.~\eqref{eq:syndromeiqpprobability} are real, the two amplitudes are complex conjugates and again $p_{\ket{\overline{0}}}^{\cC}=p_{\ket{\overline{1}}}^{\cC}$. When these distributions differ, sampling for a known encoded input reduces to their classical mixture in Eq.~\eqref{eq:inputmixture}.

A proof of Lemma~\ref{lem:syndromeiqp} is given in Appendix~\ref{app:iqpamplitudes}. For the surface-code noise of Eq.~\eqref{eq:noise}, $\cI$ consists of the singletons $\{q\}$, $q\in[n]$, and the edges $e\in\cE$. Every such support $\Gamma\in\cI$ flips at most two $X$-stabilizers (i.e. $|\partial\Gamma| \leq 2$), so $W$ consists of one- and two-qubit diagonal gates only, that is, $W=e^{iH}$ for an Ising Hamiltonian $H$ on the syndrome register whose interaction graph is determined by the code geometry. Lemma~\ref{lem:surfacecodeiqp} in Appendix~\ref{app:iqpgraphs} identifies these graphs, with the following results:
\begin{enumerate}[(i)]
\item For nearest-neighbor crosstalk $U_{\mathrm{nn}}$ on a rotated surface code $\Crotated{d}$ of odd distance $d$, $H$ is an Ising model on two disconnected $\frac{d-1}{2}\times\frac{d+1}{2}$ square lattices, one for each of the two components $A$ and $B$ of $X$-stabilizers, with local fields (corresponding to $|\partial\Gamma|=1$) on one boundary column of each lattice (Fig.~\ref{fig:pairiqp}(a)). The $d-1$ edges $e$ for which $Z(e)$ is a weight-two $Z$-stabilizer of $\Crotated{d}$ contribute only a global phase ($|\partial\Gamma|=0$).
\item For single-qubit rotations $U_{\mathrm{sq}}$ on a standard (unrotated) surface code $\Cstandard{D}$ of distance $D$, $H$ is an Ising model on a single $(D-1)\times D$ square lattice, with local fields on its left and right boundary columns.
\item For combined noise $U_{\mathrm{comb}}=U_{\mathrm{sq}}U_{\mathrm{nn}}$ on a rotated surface code, $H$ is an Ising model on a finite sublattice of the king lattice, in which every site is coupled to its horizontal, vertical, and diagonal neighbors, with local fields on the sites near the left and right boundaries (Fig.~\ref{fig:kingiqp}). 
\end{enumerate}
The king graph in (iii) is non-planar, which is the source of the hardness result in Sec.~\ref{sec:hardness}.

\section{Efficient simulation of coherent crosstalk}
\label{sec:mapping}

\begin{figure*}[!t]
\centering
\begin{minipage}[t]{191pt}
\centering
\subfloat[{Original rotated surface code, $d=7$.}]{\begin{minipage}{\linewidth}
 \centering
 \includegraphics{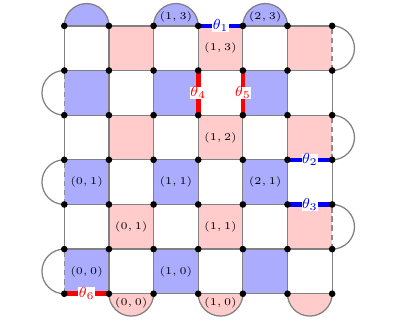}
\end{minipage}}
\end{minipage}\hfill
\begin{minipage}[t]{153pt}
\centering
\subfloat[{Auxiliary standard (unrotated) surface code $A$, $D=4$.}]{\begin{minipage}{\linewidth}
 \centering
 \includegraphics{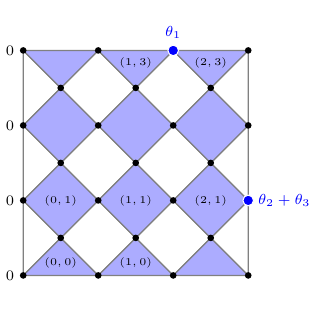}
\end{minipage}}
\end{minipage}\hfill
\begin{minipage}[t]{153pt}
\centering
\subfloat[{Auxiliary standard (unrotated) surface code $B$, $D=4$.}]{\begin{minipage}{\linewidth}
 \centering
 \includegraphics{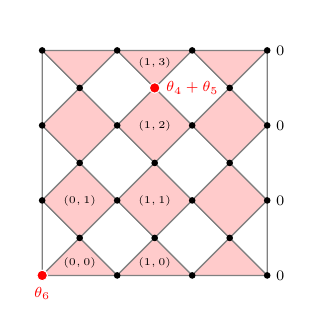}
\end{minipage}}
\end{minipage}
\caption{{Reduction of nearest-neighbor crosstalk $U_{\mathrm{nn}}$ on a rotated surface code of distance $d=7$ to single-qubit noise on two standard (unrotated) surface codes of distance $D=4$. Panel (a) shows the original rotated surface code; panels (b) and (c) show the independent auxiliary surface codes $A$ and $B$, respectively. Black dots denote physical qubits. Blue and red faces denote $X$-stabilizers of components $A$ and $B$, respectively, and white faces denote $Z$-stabilizers. An $X$ face $(u,v)$ in (a) corresponds to the face with the same color and grid label in (b) or (c). Thick blue and red bonds in (a), labeled at their centers by $\theta_k$, denote crosstalk rotations $e^{i\theta_k ZZ}$ assigned to auxiliary surface code $A$ or $B$, respectively. Grey dashed bonds are the weight-two $Z$-stabilizers on the left and right boundaries, whose rotations have trivial syndrome and contribute only a global phase. In (b) and (c), bright blue and red dots are the physical qubits of the auxiliary codes carrying the resulting single-qubit rotations $e^{i\varphi Z}$, labeled by the angle sum they receive, while the sites labeled $0$ on the inactive rough boundary of each surface code denote a physical qubit with no error (zero angle rotation). The illustrated mappings are: $\theta_1$ flips $(1,3)$ and $(2,3)$ in $A$; $\theta_2$ and $\theta_3$ each flip the single stabilizer $(2,1)$ in $A$; $\theta_4$ and $\theta_5$ each flip $(1,2)$ and $(1,3)$ in $B$; and $\theta_6$ flips the single stabilizer $(0,0)$ in $B$. A syndrome with two flipped stabilizers maps to the auxiliary qubit between the two correspondingly marked faces of the auxiliary code, and a syndrome with one flipped stabilizer maps to an auxiliary qubit on the rough boundary. Interactions with the same syndrome share an auxiliary qubit and their angles add, as in $\theta_2+\theta_3$ and $\theta_4+\theta_5$; this adds rotation angles, not Pauli error probabilities. The coordinate conventions are those of Sec.~\ref{sec:model}; Appendix~\ref{app:rotatedmapping} names the corresponding physical qubits of the auxiliary codes.}}
\label{fig:mapping}
\end{figure*}

By Sec.~\ref{sec:recovery}, simulating coherent crosstalk on a rotated surface code $\Crotated{d}$ amounts to producing a sample from  the syndrome distribution $p(s)=p^{\Crotated{d}}(s\mid U_{\mathrm{nn}})$; the conditional logical state then can be determined according to Lemma~\ref{lem:logicalaction}. We now reduce this sampling problem -- for an odd-distance-$d$ rotated surface code with nearest-neighbor crosstalk -- to the problem of producing  samples from the syndrome distributions of two independent auxiliary instances of   standard (unrotated) surface codes $\Cstandard{(d+1)/2}$ of distance~$(d+1)/2$, with single-qubit coherent noise. For the latter problem,  efficient samplers are known~\cite{BravyiEnglbrechtKoenigPeard18,VennBeri20}. The construction  relates crosstalk interaction angles in the original rotated code  to corresponding interaction angles of single-qubit $Z$-rotations in the two auxiliary standard codes; the syndrome distribution of the original code then factorizes into the product of syndrome distributions associated with the two auxiliary codes.

\subsection{Reduction to single-qubit noise}
\label{sec:reduction}

Consider a rotated surface code~$\Crotated{d}$ of odd distance~$d\geq5$. Set $D=(d+1)/2$. As shown in Fig.~\ref{fig:mapping} using two colors, the set of $X$-stabilizers of $\Crotated{d}$ is partitioned into two disjoint subsets $A$ (blue) and $B$ (red), as specified in Eqs.~\eqref{eq:bulkxchecks} and~\eqref{eq:boundaryxchecks}. Each component contains $|A|=|B|=r_X/2=D(D-1)$ $X$-stabilizers, indexed by the grid~$G$ of Eq.~\eqref{eq:componentgrid}.

Now consider two auxiliary standard (unrotated) surface codes of distance~$D$, each with $m=D^2+(D-1)^2$ physical qubits, which we denote by $\Caux{A}$ and $\Caux{B}$. By Eq.~\eqref{eq:codesizes}, $\Cstandard{D}$ has $D(D-1)$ $X$-stabilizers, and we index them by the same grid~$G$, as in Fig.~\ref{fig:mapping}(b),(c). This gives a bijection between the $X$-stabilizers of $\Crotated{d}$ and the disjoint union of the $X$-stabilizers of $\Caux{A}$ and $\Caux{B}$: the stabilizer $S^{X,(C)}_{u,v}$ of component $C\in\{A,B\}$ is identified with the stabilizer $\widetilde S^{X,(C)}_{u,v}$ of $\Caux{C}$ carrying the same grid index $(u,v)\in G$.

We write the syndrome as $s=(s^{(A)},s^{(B)})\in\{0,1\}^{r_X}$, where $s^{(C)}=(s^{(C)}_{u,v})_{(u,v)\in G}\in\{0,1\}^{r_X/2}$ collects the outcomes of component~$C$ and, under this bijection, is also a syndrome of $\Caux{C}$. 

Let $\Vrotated{d}$ be the set of physical qubits of $\Crotated{d}$ and $\Erotated{d}=\cE$ the set of edges, associated with nearest-neighbor $ZZ$-interactions of the form~$Z(e)$, $e\in \Erotated{d}$ appearing in Eq.~\eqref{eq:noise}. The pair $(\Vrotated{d},\Erotated{d})$ forms the $d\times d$-grid graph of Fig.~\ref{fig:mapping}(a). For an edge $e=\{q,q'\}\in\Erotated{d}$, we write $\partial e\in\{0,1\}^{r_X}$ for the syndrome of the error $Z(e)=Z_qZ_{q'}$ and $(\partial e)^{(C)}$ for its restriction to component~$C$. 

 In more detail, the error $Z(e)$ anticommutes with an $X$-stabilizer if and only if the corresponding face (defining the $X$-stabilizer) contains exactly one of the two qubits of~$e$. For a horizontal edge~$e$, these $X$-stabilizers correspond to the  faces located to its left and to its right. For a vertical edge~$e$, the error $Z(e)$ affects the $X$-stabilizers at the faces  below and above it. In either case both faces belong to the same component $C\in\{A,B\}$. Hence there are three kinds of edges~$e$ with different effects of the associated two-qubit operator~$Z(e)$. For an edge~$e$ in the bulk, the operator~$Z(e)$ changes the eigenvalue of two stabilizers of one component~$C$, whose faces are adjacent in the grid~$G$, that is, $(\partial e)^{(C)}$ has support $\{(u,v),(u+1,v)\}$ or $\{(u,v),(u,v+1)\}$, and $(\partial e)^{(\bar C)}=0$, where $\bar C\neq C$ denotes the other component. For a  horizontal edge~$e$ adjacent to the left or right boundary, the operator~$Z(e)$ flips a single $X$-stabilizer, which lies in the right column $u=D-2$ of component~$A$ or in the left column $u=0$ of~$B$. Finally, the $d-1$~vertical edges~$e$ on the left and right boundaries for which the operator~$Z(e)$ is a weight-two $Z$-stabilizer (dashed in Fig.~\ref{fig:mapping}(a)) have trivial syndrome. We write $\Erotatedstar{d}\subset\Erotated{d}$ for the set of edges of the first two kinds; these are exactly the edges~$e$ with $\partial e\neq0$.

We now consider the two standard surface codes~$\Caux{C}$, where $C\in \{A,B\}$. We write $\partial^{(C)}\{q\}\in\{0,1\}^{G}$ for the syndrome of  the error~$Z_q$ on a qubit $q\in\Vstandard{D}$ of $\Caux{C}$. A single-qubit error $Z_q$ in the bulk changes the eigenvalue of  the two $X$-stabilizers whose faces contain~$q$, which are adjacent in the grid~$G$. On the  other hand, a single-qubit error~$Z_q$ on any of the $D$ qubits on each left or right boundary, changes only the eigenvalue of the unique stabilizer in the left column $u=0$ or the right column $u=D-2$ whose face contains~$q$.

We note that for every $e\in\Erotatedstar{d}$ there is a unique component $C$ and a unique qubit $q\in\Caux{C}$  such that a $Z$-error on $q$ in $\Caux{C}$ flips exactly the stabilizers flipped by $Z(e)$ in the rotated surface code (under the bijection of syndromes discussed above). This defines a map
\begin{align}
\begin{matrix}
 \iota\colon&\Erotatedstar{d}&\to&\{A,B\}\times\Vstandard{D}\ ,\\
 & e& \mapsto & \iota(e)=(C,q)
 \end{matrix}
 \label{eq:pairgroups}
 \end{align}
 We note that for $\iota(e)=(C,q)$ this map satisfies 
 \begin{align}
 (\partial e)^{(C)}=\partial^{(C)}\{q\}\neq0\ ,\\
 (\partial e)^{(\bar C)}=0\ .
\end{align}

The map $\iota$ is not injective, since edges with the same syndrome have the same image (such as those carrying the labels $\theta_2$ and $\theta_3$ in Fig.~\ref{fig:mapping}). Its image consists of all qubits of $\Caux{A}$ and $\Caux{B}$ except those on the left rough boundary of $\Caux{A}$ and on the right rough boundary of $\Caux{B}$, and every qubit in the image is the image of one or two edges; Appendix~\ref{app:rotatedmapping} gives $\iota$ explicitly, Eq.~\eqref{eq:auxangles}.

Given crosstalk angles $\bm\theta=(\theta_e)_{e\in\Erotated{d}}$, we assign to each qubit of an auxiliary code the sum of the angles associated with the edges mapped to it, and subject the auxiliary codes to the resulting single-qubit rotations:
\begin{align}
 \varphi_q^{(C)}=\sum_{e\in\cE_{d}^{\mathrm{rot}, *}: \iota(e)=(C,q)}\theta_e\ ,
 \qquad
 U^{(C)}=\prod_{q\in\Vstandard{D}}e^{i\varphi_q^{(C)}Z_q^{(C)}}\ ,
 \label{eq:anglesum}
\end{align}
where $Z^{(C)}_q$ acts on qubit $q$ of $\Caux{C}$. A qubit outside the image of $\iota$ receives angle zero. The following lemma states the reduction that results from this construction.
\begin{lemma}[Reduction to single-qubit noise]
\label{lem:mapping}
Let $d\geq5$ be odd and $D=(d+1)/2$. Let $U_{\mathrm{nn}}=U_{\mathrm{nn}}(\bm\theta)$ be coherent nearest-neighbor $ZZ$-crosstalk noise, Eq.~\eqref{eq:noise}, on the rotated surface code $\Crotated{d}$, with arbitrary angles $\bm\theta$. Let $U^{(A)}$ and $U^{(B)}$ be the single-qubit noise~\eqref{eq:anglesum} on the auxiliary standard surface codes $\Caux{A}$ and $\Caux{B}$. Then, for every syndrome $s=(s^{(A)},s^{(B)})\in\{0,1\}^{r_X}$,
\begin{align}
 p^{\Crotated{d}}(s\mid U_{\mathrm{nn}})=p^{\Cstandard{D}}(s^{(A)}\mid U^{(A)})\,p^{\Cstandard{D}}(s^{(B)}\mid U^{(B)})\ ,
 \label{eq:factor}
\end{align}
where none of the three distributions depends on the encoded state.
\end{lemma}
\noindent The proof of this lemma can be found in Appendix~\ref{app:mappingproof}. It also explains why the noise considered here makes all three distributions state-independent. Thus the original syndrome can be sampled by independently sampling the two auxiliary syndromes and combining their outcomes. The edges~$e$ for which $Z(e)$ has trivial syndrome ($\partial e=0$) do not affect syndrome probabilities $p^{\Crotated{d}}(s\mid U_{\mathrm{nn}})$, which is why they are omitted from Eq.~\eqref{eq:anglesum}.

\subsection{Sampling algorithm}
\label{sec:sampling}

For single-qubit coherent noise, efficient algorithms for sampling a syndrome are known for rotated codes~\cite{BravyiEnglbrechtKoenigPeard18} and standard codes~\cite{VennBeri20}. We write $\Sample_{\mathrm{sq}}^{\Cstandard{D}}(\bm\varphi)$ for the standard-code subroutine used below, which returns a syndrome $s$ distributed with probability $p^{\Cstandard{D}}(s\mid U_{\mathrm{sq}}(\bm\varphi))$.

The reduction expressed by Lemma~\ref{lem:mapping} gives the following syndrome sampling algorithm for nearest-neighbor coherent crosstalk on a rotated surface code.

\begin{center}
\fbox{\begin{minipage}{0.94\columnwidth}
\small
\textbf{Algorithm 1: $\Sample_{\mathrm{nn}}^{\Crotated{d}}(\bm\theta)$}\par
\textbf{(Syndrome sampling for nearest-neighbor coherent crosstalk)}\par
\emph{Input:} an odd distance $d\geq5$ and real angles $\bm\theta=(\theta_e)_{e\in\cE}$ for the edges of the rotated code $\Crotated{d}$.\par
\begin{enumerate}
 \setlength{\itemsep}{1pt}
 \setlength{\parskip}{0pt}
 \item Initialize the single-qubit error angle $\varphi_q^{(C)}=0$ for every physical qubit $q$ of each auxiliary code $C\in\{A,B\}$ of distance $D=(d+1)/2$. This angle specifies the rotation $e^{i\varphi_q^{(C)}Z_q^{(C)}}$.
 \item For each edge $e=\{q,q'\}\in\cE_d^{\mathrm{rot}}$, compute $\partial e\in\{0,1\}^{\{A,B\}\times G}$, the syndrome of the error $Z_qZ_{q'}$. Its nonzero entries label the flipped $X$-stabilizers, determined by Eqs.~\eqref{eq:bulkxchecks} and \eqref{eq:boundaryxchecks}. If $\partial e=0$, continue to the next edge. Otherwise, identify its component $C$ and the unique qubit $a$ of auxiliary code $\Caux{C}$ (i.e. compute $\iota(e)=(C,a)$) with $\partial^{(C)}\{a\}=(\partial e)^{(C)}$, as in Eq.~\eqref{eq:pairgroups}, and add $\theta_e$ to $\varphi_a^{(C)}$.
 \item Sample the two syndromes independently:
 \[
 \begin{aligned}
 s^{(A)}&\leftarrow\Sample_{\mathrm{sq}}^{\Cstandard{D}}(\bm\varphi^{(A)})\ ,\\
 s^{(B)}&\leftarrow\Sample_{\mathrm{sq}}^{\Cstandard{D}}(\bm\varphi^{(B)})\ .
 \end{aligned}
 \]
 \item Return the syndrome $s=(s^{(A)},s^{(B)})$ of the original code by concatenating the outcomes along the labels $\{A,B\}\times G$.
\end{enumerate}
\emph{Output:} $s\sim p^{\Crotated{d}}(\,\cdot\mid U_{\mathrm{nn}}(\bm\theta))$.
\end{minipage}}
\end{center}
The correctness of the algorithm, i.e., the fact that its output $s$ is distributed according to $p^{\Crotated{d}}(\,\cdot\mid U_{\mathrm{nn}}(\bm\theta))$, follows from Lemma~\ref{lem:mapping}.

Given a sampled syndrome $s$, a decoder is a map
\begin{align}
\begin{matrix}
 \mathsf{Dec}:&\{0,1\}^{r_X}&\rightarrow &\{\Gamma:\Gamma\subseteq[n]\}\\
 & s & \mapsto &\Gamma_s
 \end{matrix}
 \end{align}
with $\partial\Gamma_s=s$; its correction is $C_s=Z(\Gamma_s)$, applying $Z$ to the qubits in $\Gamma_s$. To determine the conditional logical action when using such a decoder, compute the support $\Gamma_s=\mathsf{Dec}(s)$ and its parity $t_s=|\Gamma_s|\bmod2$. By Lemma~\ref{lem:logicalaction}, $t_s=0$ means the identity and $t_s=1$ means it is a logical $\logicalZ$, equivalently a logical rotation by $\Theta_s=\pi t_s/2$ up to a global phase. 

\subsection{Runtime of algorithm}
\label{sec:cost}
A rotated surface code of distance $d$ maps to two auxiliary standard (unrotated) surface codes of distance $D=(d+1)/2$. Computation of angles, and concatenation of the output samples take $O(d^2)$ operations; the two calls to the single-qubit syndrome sampler dominate the complexity. Write $T(\mathsf A)$ for the number of elementary arithmetic operations required for one call to a sampler $\mathsf A$.

The standard-code subroutine uses fermionic linear optics (FLO). A distance-$D$ standard surface code $\Cstandard{D}$ has $D^2+(D-1)^2=O(D^2)$ physical qubits; a syndrome sample requires $O(D^2)$ elementary FLO operations, each realizable with $O(D^4)$ arithmetic operations~\cite{VennBeri20}. Hence
\begin{align}
 T\!\left(\Sample_{\mathrm{sq}}^{\Cstandard{D}}\right)=O(D^6)\ ,
 \label{eq:stdcost}
\end{align}
and the runtime per crosstalk syndrome sample is
\begin{align}
 T\!\left(\Sample_{\mathrm{nn}}^{\Crotated{d}}\right)
 =O(d^2+D^6)=O(d^6)\ .
 \label{eq:maincost}
\end{align}

Algorithm~1 and the identities underlying it are formulated for arbitrary real
angles and exact arithmetic. A numerical implementation must fix how the angles are
represented with machine precision. We distinguish two cases. If
the numbers $\cos\theta_e$ and $\sin\theta_e$ are given exactly as elements
of a fixed algebraic number field $K$ (rationals, or, as for the angles
$k\pi/8$ of Theorem~\ref{thm:hardness}, elements of
$\mathbb Q(\cos\frac{\pi}{8})$), with at most $b$ bits per field element, we can carry out all arithmetic exactly in $K$.
Every quantity the algorithm computes, from the entries of the fermionic
covariance matrix to the conditional probabilities from which the syndrome
bits are drawn, is a ratio of Pfaffians of matrices with entries in $K$ of
polynomial size, so all intermediate numbers have $\poly(d,b)$ bits and the
$O(d^6)$ arithmetic operations of Eq.~\eqref{eq:maincost} cost $\poly(d,b)$ bit
operations. Exact sampling additionally compares a uniform random number, revealed through
independent fair bits, with successively refined certified bounds on each conditional probability;
this takes expected polynomial bit time, rather than a fixed worst-case bound.
The output is distributed exactly as
$p^{\Crotated{d}}(\,\cdot\mid U_{\mathrm{nn}}(\bm\theta))$; in particular
the multiplicative criterion~\eqref{eq:introapprox} holds with $\epsilon=0$.
If instead the angles are arbitrary reals with polynomial-time access to their binary expansions, fix $\delta>0$, set
$b_{\mathrm{prec}}=\lceil\log_2(4d^2/\delta)\rceil$, and replace each $\theta_e$ by an angle
$\tilde\theta_e$ with rational $\cos\tilde\theta_e,\sin\tilde\theta_e$ of
$O(b_{\mathrm{prec}})$ bits and $|\theta_e-\tilde\theta_e|\leq2^{-b_{\mathrm{prec}}}$. Since each factor of
$U_{\mathrm{nn}}$ moves by at most $2^{-b_{\mathrm{prec}}}$ in operator norm and there are
fewer than $2d^2$ factors,
$\|U_{\mathrm{nn}}(\bm\theta)-U_{\mathrm{nn}}(\tilde{\bm\theta})\|\leq2d^2\,2^{-b_{\mathrm{prec}}}$,
and the exact sampler for $\tilde{\bm\theta}$ produces a distribution
$\tilde p$ with $|\tilde p(s)-p(s)|\leq\delta$ for every syndrome and
$\|\tilde p-p\|_1\leq\delta$, at expected bit cost $\poly(d,\log(1/\delta))$. No
multiplicative guarantee relative to $\bm\theta$ itself can be expected in
this case, since $p$ may vanish by interference at isolated angles.

\section{Numerical results}
\label{sec:numerics}

\begin{figure*}[!tbp]
\centering
\makebox[0.8\textwidth][c]{
\begin{minipage}[t]{0.42\textwidth}
\centering
\subfloat[{$P^L$ under coherent crosstalk noise $U_{\mathrm{nn}}$.}]{\begin{minipage}{\linewidth}
\centering
\includegraphics[width=\linewidth]{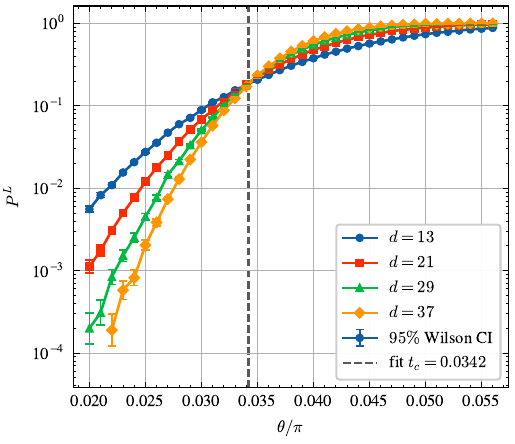}
\end{minipage}}
\end{minipage}\hspace{0.04\textwidth}
\begin{minipage}[t]{0.42\textwidth}
\centering
\subfloat[{$P^L_{\mathrm{twirl}}$ under the Pauli twirl $\mathcal{N}_{\mathrm{twirl}}$.}]{\begin{minipage}{\linewidth}
\centering
\includegraphics[width=\linewidth]{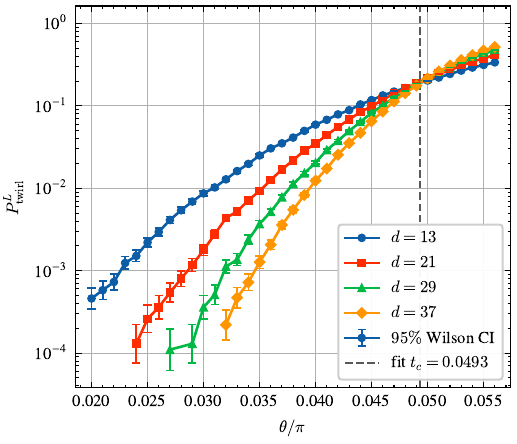}
\end{minipage}}
\end{minipage}}
\caption{Standard-MWPM logical error on odd-distance rotated codes for (a) coherent nearest-neighbor $ZZ$ rotations and (b) the per-interaction physical Pauli twirl of Eq.~\eqref{eq:physicaltwirl}. The vertical axes are logarithmic; $d$ denotes code distance. Each point uses $200\,000$ samples. Both panels cover the same sampled range, $0.020\leq\theta/\pi\leq0.056$, in steps of $0.001$. Only estimates strictly greater than $10^{-4}$ are plotted; lines connect the displayed estimates. The cutoff is for display only and does not imply zero logical error at omitted points. The normalization assigns $P^L=2$ to a certain logical Pauli error, so $P^L=1$ corresponds to failure probability $1/2$. Error bars show pointwise 95\% Wilson confidence intervals for Monte Carlo sampling uncertainty. Dashed lines mark threshold estimates from Table~\ref{tab:threshold}.}
\label{fig:numerics}
\end{figure*}

Minimum-weight perfect matching (MWPM) \cite{Edmonds1965,FowlerWhitesideHollenberg2012} is a common surface-code decoder and provides a useful benchmark when recovery is not tailored to the physical noise model. For the coherent crosstalk $U_{\mathrm{nn}}$ studied here, Eq.~\eqref{eq:noise}, a correction constrained to have even weight recovers exactly at every interaction strength, as shown in Sec.~\ref{sec:recovery}. The conventional MWPM decoder studied here can instead return odd-weight corrections and thereby introduce a logical $Z$-error. We therefore quantify the residual logical error it leaves under coherent crosstalk and compare it with that of the corresponding incoherent model.

We compare uniform coherent nearest-neighbor $ZZ$ noise $U_{\mathrm{nn}}$ with a physical Pauli approximation on rotated surface codes of distance $d\in\{13,21,29,37\}$. Preparation, syndrome measurements, and recovery are assumed to be ideal. For both noise models and every distance, we look at the same $37$ rotation angles
\begin{align}
 \begin{split}
 \theta/\pi\in{}&\{0.020+0.001j:j=0,\ldots,36\}\ .
 \end{split}
 \label{eq:numericalgrid}
\end{align}
Each distance, angle, and noise model uses $N=200\,000$ syndrome samples of the original code, giving $7.4\times10^6$ samples per distance and model and $5.92\times10^7$ samples in total. Figure~\ref{fig:numerics} shows diamond-norm errors for both models over this range. Syndromes for coherent crosstalk are generated by the reduction described in Sec.~\ref{sec:mapping}, with auxiliary surface code distances $D=(d+1)/2\in\{7,11,15,19\}$.

\subsection{Logical error, physical twirling, and recovery}

For fixed $d$ and $\theta$, write
$p(s)=p^{\Crotated{d}}(s\mid U_{\mathrm{nn}}(\theta))$, where
$U_{\mathrm{nn}}(\theta)$ denotes uniform angles $\theta_e=\theta$ on every edge $e$. We use the
average conditional diamond-norm error,
\begin{align}
 P^L=\sum_{s\in\{0,1\}^{r_X}} p(s)\,\|\Lambda_s-\mathrm{id}\|_\diamond\ ,
 \label{eq:logicalmetric}
\end{align}
where $\mathrm{id}$ is the (logical) identity channel and $\Lambda_s$ is the
normalized conditional logical channel
\begin{align}
\Lambda_s:\ \rho\mapsto C_s\Pi_sU_{\mathrm{nn}}\rho\,
 U_{\mathrm{nn}}^\dagger\Pi_sC_s^\dagger/p(s)
\end{align}
restricted to density operators supported on the code space. 
The map $\Lambda_s$ is linear because of the input independence of the syndrome distribution. Note that $P^L\in[0,2]$ is an error metric, not a probability: a definite logical Pauli-$Z$ error gives $P^L=2$.

For the crosstalk noise studied here, Eq.~\eqref{eq:parity} makes
every conditional channel $\Lambda_s$ either the identity or conjugation by $\logicalZ$,
according to the weight parity of the correction $C_s$
returned by the decoder. Hence
\begin{align}
 P^L=2p_{\mathrm{fail}}^L\ ,
 \qquad p_{\mathrm{fail}}^L=\Pr_{s\sim p}[\,|\mathsf{supp}(C_s)|\text{ odd}\,]\ .
 \label{eq:oddmetric}
\end{align}
In particular, averaging over the syndromes gives the averaged logical channel
\begin{align}
 \Lambda(\rho)=(1-p_{\mathrm{fail}}^L)\rho
                 +p_{\mathrm{fail}}^L \logicalZ\rho \logicalZ\ ,
 \label{eq:rotatedaverage}
\end{align}
a dephasing channel. Note that $P^L=\|\Lambda-\mathrm{id}\|_\diamond$. This last
identity is special to the considered coherent crosstalk noise and need not hold for
general coherent noise.

For comparison, we consider incoherent noise that replaces the rotation on each edge
$e=\{q,q'\}\in\cE$ separately by its Pauli twirl,
\begin{align}
 \cN_e^{\mathrm{twirl}}(\rho)
 =\cos^2\theta\,\rho+\sin^2\theta\,Z_qZ_{q'}\rho Z_qZ_{q'}\ ,
 \label{eq:physicaltwirl}
\end{align}
and composes these channels over all $2d(d-1)$ edges
$e\in\cE$,
\begin{align}
 \cN_{\mathrm{twirl}}
 =\cN_{e_1}^{\mathrm{twirl}}\circ\cdots
  \circ\cN_{e_{2d(d-1)}}^{\mathrm{twirl}}\ ,
 \label{eq:twirlcomposition}
\end{align}
which is order-independent because these $Z$-type Pauli channels commute. Equivalently, the $ZZ$-error $Z(e)$ is applied independently on each edge $e$ with probability $\sin^2\theta$. This twirls each interaction, not the complete product unitary. Its syndrome distribution $p_{\mathrm{twirl}}$ is state-independent, and every sampled error has even weight, so Eqs.~\eqref{eq:oddmetric} and~\eqref{eq:rotatedaverage} apply verbatim with $p$ replaced by $p_{\mathrm{twirl}}$. We write $P^L_{\mathrm{twirl}}$ for the resulting metric.

For each distance, angle, and coherent / incoherent (twirled) model we record the number $M$ of the $N$ sampled syndromes whose MWPM correction has odd weight, and report the empirical frequency
\begin{align}
 \hat f=M/N\ \qquad \text{and} \qquad \widehat P^L=2\hat f\ .
 \label{eq:estimator}
\end{align}

\subsection{Coherent versus incoherent errors}

Figure~\ref{fig:numerics} shows $\widehat P^L$ for both models over the full range of angles. At equal angle, MWPM leaves a substantially larger logical error under coherent crosstalk than under its twirl: at $d=37$ and $\theta/\pi=0.040$ we find $\widehat P^L=0.611(2)$ against $\widehat P^L_{\mathrm{twirl}}=0.0124(4)$, a ratio of about $49$. The parentheses in these two estimates denote one binomial standard error.

To summarize how the logical error depends on distance and angle, we perform a finite-size scaling fit. Near a threshold, the failure rate is expected to depend on distance and angle only through the single scaling variable $x=(t-t_c)\,d^{1/\nu}$, where $t=\theta/\pi$, $t_c$ is a location parameter that may be interpreted as a threshold if the scaling ansatz holds, and $\nu$ is a critical exponent. Truncating the dependence on $x$ at second order gives the standard fitting form~\cite{WangHarringtonPreskill03,TuckettBartlettFlammia18}
\begin{align}
 f(t,d)\simeq A+Bx+Cx^2\ ,\qquad x=(t-t_c)\,d^{1/\nu}\ ,
 \label{eq:thresholdscaling}
\end{align}
where $f=p_{\mathrm{fail}}^L$ is the odd-correction probability of Eq.~\eqref{eq:oddmetric} and $\hat f$ of Eq.~\eqref{eq:estimator} is its estimator. For each noise model we fit the five parameters $t_c,\nu,A,B,C$ once to all four distances simultaneously, with no correction-to-scaling term. Fit windows, weights, and the jackknife spread are specified in Appendix~\ref{app:thresholdfit}, where Fig.~\ref{fig:thresholdfit} also shows the resulting collapse of all four distances onto a single curve. Table~\ref{tab:threshold} collects the fitted values.

\begin{table}[!ht]
\centering
\begin{tabular}{lccc}
\toprule
Noise model & $t_c$ & $\nu$ & $\chi^2_{\mathrm{red}}$ \\
\midrule
Coherent       & $0.0342(4)$ & $1.56$ & $7.10$ \\
Physical twirl & $0.0493(5)$ & $1.49$ & $6.50$ \\
\bottomrule
\end{tabular}
\caption{Parameters of the joint fit of Eq.~\eqref{eq:thresholdscaling} to the four simulated distances. Parentheses give the leave-one-distance-out spread $s_{\mathrm{JK}}$ of Appendix~\ref{app:thresholdfit} in the last digit. Each fit has $23$ degrees of freedom.}
\label{tab:threshold}
\end{table}

The fitted $t_c$ is about $30\%$ smaller for coherent noise than for its twirl, and this ordering survives deleting any single distance or changing either fit window (Appendix~\ref{app:thresholdfit}). The absolute values, however, deserve caution: the reduced $\chi^2$-values of $7.10$ and $6.50$ show that the quadratic ansatz does not describe the data within sampling uncertainty. We therefore read Table~\ref{tab:threshold} as a set of model-dependent extrapolations rather than as controlled asymptotic thresholds, and note that $s_{\mathrm{JK}}$ covers only the set of distances included.

\subsection{Runtime}

Table~\ref{tab:runtime} reports the average per-sample cost of syndrome generation and of decoding. For coherent noise, sampling dominates decoding at every simulated distance; for the twirl it is two to three orders of magnitude cheaper, and decoding dominates throughout.

\begin{table}[!htbp]
\centering
\begin{tabular}{rrr}
\toprule
$d$ & Sampling ($\mu$s) & MWPM ($\mu$s) \\
\midrule
\multicolumn{3}{c}{Coherent crosstalk} \\
\midrule
13 & 175.800 & 24.549 \\
21 & 759.948 & 137.801 \\
29 & 2208.120 & 494.955 \\
37 & 5101.692 & 1362.194 \\
\midrule
\multicolumn{3}{c}{Physical Pauli twirl} \\
\midrule
13 & 1.477 & 10.330 \\
21 & 3.716 & 44.529 \\
29 & 7.043 & 143.966 \\
37 & 11.439 & 369.166 \\
\bottomrule
\end{tabular}
\caption{Measured stage costs on an Apple M2 Max using serial execution on odd-distance rotated surface codes. Each row is a sample-weighted mean over all $37$ angles of Eq.~\eqref{eq:numericalgrid}, with $200\,000$ samples per angle ($7.4\times10^6$ per row). Entries are microseconds per sample of the original code. For coherent noise, sampling includes both auxiliary code samples and the mapping to the $X$-stabilizer syndrome of the original code; for the physical twirl, it includes independent draws of the $ZZ$-error on each edge, combining their effects on each qubit, and constructing the syndrome. MWPM includes matching, correction construction, and validation. Setup and other wall-time overheads are excluded.}
\label{tab:runtime}
\end{table}

\section{Limits of efficient simulation: combined noise}
\label{sec:hardness}

The reduction of Sec.~\ref{sec:mapping} relies on the splitting of the syndromes of $ZZ$-errors into two disconnected sets. Single-qubit rotations connect them, and the resulting sampling problem is no longer known to be efficiently solvable. To state this precisely, take $\cC=\Crotated{d}$, $U=U_{\mathrm{comb}}(\bm\varphi,\bm\theta)$ of Eq.~\eqref{eq:noise}, and fix the encoded input to $\ket{\overline{0}}$. Combined noise can have an input-dependent syndrome distribution, so we keep the encoded state in the notation, writing $p(s)=p_{\ket{\overline{0}}}^{\cC}(s\mid U)$. Lemma~\ref{lem:syndromeiqp} directly gives
\begin{align}
 p_{\ket{\overline{0}}}(s)&=\left|\bra{s}\HH^{\otimes r_X}W\ket{+}^{\otimes r_X}\right|^2\ ,
\end{align}
where
\begin{align}
 W&=\prod_q e^{i\varphi_qZ(\partial\{q\})}
       \prod_{e\in\cE}e^{i\theta_eZ(\partial e)}\ .
 \label{eq:iqp}
\end{align}
Here $\HH$ is the Hadamard gate, $\ket{+}=\HH\ket0$, and
$Z(\delta)=\prod_{u=1}^{r_X}Z_u^{\delta_u}$ acts on the syndrome register (i.e., $(\mathbb{C}^2)^{\otimes r_X}$) rather than on the physical qubits, with $\delta\in\{0,1\}^{r_X}$ a syndrome vector. Eq.~\eqref{eq:iqp} is an instantaneous quantum polynomial-time (IQP) circuit: commuting diagonal gates act on $\ket{+}^{\otimes r_X}$ and every qubit is then measured in the eigenbasis of $X$. Because $|\partial\{q\}|\leq2$ and $|\partial e|\leq2$, each gate in the product defining $W$ couples at most two ``syndrome'' qubits.

To see explicitly that combined noise can give $p_{\ket{\overline{0}}}\neq p_{\ket{\overline{1}}}$, on $\Crotated{5}$ take angle $\pi/4$ on the four supports $\{q_{0,0},q_{1,0}\}$, $\{q_{2,0}\}$, $\{q_{3,0}\}$, and $\{q_{4,0}\}$, and zero elsewhere. Their product is the logical-$Z$ string along the bottom row. Only the empty and full selections of these supports have zero syndrome, giving $p_{\ket{\overline{0}}}(0)=1/4$ and $p_{\ket{\overline{1}}}(0)=0$.

The geometry of these couplings in $W$ is a king lattice, that is, the integer grid in which each site is adjacent to its horizontal, vertical, and diagonal neighbors. Label an $X$-stabilizer by the geometric coordinates $(x,y)$ of the lower-left corner of its face, extended to the boundary faces, so that $0\leq x\leq d-2$, $-1\leq y\leq d-1$ and $x+y$ is even; the bottom boundary stabilizers sit at $y=-1$ and the top ones at $y=d-1$. These geometric coordinates are related to the component labels $(C,u,v)$ of Eqs.~\eqref{eq:bulkxchecks} and~\eqref{eq:boundaryxchecks} by the map $\rho$ of Eq.~\eqref{eq:geocomponents}. A single-qubit rotation flips two stabilizers separated by $(\Delta x,\Delta y)=(1,\pm1)$, or a single one at the left or right boundary; a nearest-neighbor rotation flips two separated by $(2,0)$ or $(0,2)$, a single one, or none. In the coordinates
\begin{align}
 (a,b)=\left(\frac{x+y}{2},\frac{x-y}{2}\right)\ ,
 \label{eq:kingcoords}
\end{align}
the two-stabilizer supports of the first kind are the axial edges and those of the second kind are the diagonal edges of a king lattice; the remaining supports are single-site fields at the boundary. Distinct edges have disjoint sets of physical preimages, so in Eq.~\eqref{eq:iqp} their angles can be chosen independently. Appendix~\ref{app:hardness} shows that king-lattice circuits of this form, with angles in $\{k\pi/8:k=0,1,\ldots,15\}$ and at most four nonzero couplings at each site, are universal under postselection, and that unused king-lattice edges have coupling zero.

Our second main result is the following.
\begin{theorem}[Hardness of syndrome sampling]
\label{thm:hardness}
Consider rotated surface codes $\Crotated{d}$ of odd distance $d$ prepared in $\ket{\overline{0}}$ and subjected to the combined noise $U_{\mathrm{comb}}(\bm\varphi,\bm\theta)$ of Eq.~\eqref{eq:noise}. Write $p(s)=p_{\ket{\overline{0}}}^{\Crotated{d}}(s\mid U_{\mathrm{comb}}(\bm\varphi,\bm\theta))$. Fix $0\leq\epsilon<1$. Suppose there is a classical algorithm which, given $d$ and collections of angles $\bm\varphi,\bm\theta$ with all entries in $\{k\pi/8:k=0,1,\ldots,15\}$ and at most four nonzero couplings at each site of Eq.~\eqref{eq:iqp}, runs in time polynomial in $d$ and outputs a sample from a distribution $\widetilde p$ satisfying
\begin{align}
 |\widetilde p(s)-p(s)|\leq\epsilon\, p(s)\qquad\text{for every }s\ .
 \label{eq:mult}
\end{align}
Then the polynomial hierarchy collapses to its third level:
\begin{align}
 \PH\subseteq\BPP^{\NP}\subseteq\Sigma_3^p\cap\Pi_3^p\ .
 \label{eq:collapse}
\end{align}
\end{theorem}
The proof is given in Appendix~\ref{app:hardness}. It is the postselection argument of Refs.~\cite{Aaronson2005,BremnerJozsaShepherd2011} applied to the king-lattice construction: Eq.~\eqref{eq:mult} distorts any conditional probability by at most $(1+\epsilon)/(1-\epsilon)$, which amplification of the source computation absorbs. The statement concerns independently programmable angles and pointwise multiplicative error, including exact preservation of zero probabilities. It does not address approximate sampling at constant total-variation distance (which typically requires additional complexity-theoretic assumptions)~\cite{BremnerMontanaroShepherd2016}, nor spatially uniform angles $\varphi_q=\varphi$, $\theta_e=\theta$. For $\bm\varphi=0$, the regime studied in Sec.~\ref{sec:numerics}, Algorithm~1 samples efficiently.

\section{Discussion}
\label{sec:discussion}

We have shown that nearest-neighbor coherent crosstalk on a rotated surface code of odd distance $d$ can be simulated exactly using two instances of the known efficient single-qubit noise samplers~\cite{BravyiEnglbrechtKoenigPeard18,VennBeri20}, at a cost of $O(d^6)$ arithmetic operations per syndrome sample.

On odd-distance rotated codes, crosstalk noise consisting of nearest-neighbor $ZZ$-interactions is exactly correctable, both for coherent and incoherent noise. Our numerical results demonstrate that the standard MWPM decoder instead leaves a nonzero logical Pauli error whose probability depends strongly on the physical model: coherent interference can substantially increase the logical error at the same angle compared to the Pauli twirled counterpart. Joint finite-size scaling fits extrapolate a lower threshold angle for coherent noise, $\theta_c/\pi=0.0342(4)$ versus $0.0493(5)$ for its twirl. This ordering is stable under the distance and fit-window checks performed here, but the poor goodness of fit makes the absolute thresholds model-dependent extrapolations.

Adding independently programmable single-qubit rotations connects the two sets of $X$-stabilizers, identifies the syndrome probabilities with the output probabilities of a king-lattice IQP circuit, and permits postselected universal computation. This results
in a complexity-theoretic obstruction to efficient sampling. We do not know whether the uniform weak-noise regime $\varphi_q=\varphi$, $\theta_e=\theta$ of the combined model admits an efficient sampler, or whether approximate sampling in total-variation distance is computationally hard.

Several open problems remain. The reduction of Sec.~\ref{sec:mapping} uses ideal syndrome measurements; 
dealing with measurement errors is a problem of immediate practical relevance. The parity structure of Lemma~\ref{lem:logicalaction} shows that a decoder aware of the noise model removes the logical error entirely, so the more relevant question for realistic devices is how coherent crosstalk interacts with additional noise components that break this parity, such as single-qubit rotations or stochastic errors.  

\begin{acknowledgments}
RK acknowledges funding from the European Research Council under Grant Agreement No.~101001976 (project EQUIPTNT) and from the Munich Quantum Valley, which is supported by the Bavarian state government with funds from the Hightech Agenda Bayern Plus. RK would like to thank the Isaac Newton Institute for Mathematical Sciences, Cambridge, for support and hospitality during the programme ``Mathematics of many-body entanglement'' where work on this paper was undertaken; this work was supported by EPSRC grant no EP/Z000580/1.
\end{acknowledgments}

\appendix
\section{Syndrome amplitudes and coordinate conventions}
\label{app:amplitudes}

\subsection{Proof of Lemma~\ref{lem:logicalaction}}
\label{app:parityproof}
In this proof $\cC=\Crotated{d}$ and $p_{\ket{\logicalPsi}}(s)=p_{\ket{\logicalPsi}}^{\cC}(s\mid U_{\mathrm{nn}})$; the lemma shows in particular that $p_{\ket{\logicalPsi}}$ does not depend on $\ket{\logicalPsi}$, so it equals the main-text $p^{\cC}(s\mid U_{\mathrm{nn}})$.

\emph{Proof.} The argument is similar to that of Lemma~4 in the Supplementary Material of Ref.~\cite{BravyiEnglbrechtKoenigPeard18}, which treats single-qubit rotations.
It follows from the fact that every Pauli term appearing in an (expansion of) the noise~$U$ has even ($Z$-)weight; it thus applies, in particular, to nearest-neighbor coherent $ZZ$-rotations.

On a rotated surface code with odd distance~$d$, we can consider the
logical $\overline{X}$-representative
\begin{align}
 \logicalX=X^{\otimes n}
 \label{eq:alldataX}
\end{align}
which acts on all qubits.
Then
\begin{align}
\logicalX C_s\logicalX=(-1)^{t_s}C_s\label{eq:logicalXcsparity}
\end{align} by definition; furthermore, every even-weight $Z$-type operator commutes with~$\logicalX$.

Let $P_{\cC}$ denote the projector onto the full code space (including the $Z$-stabilizer constraints), distinct from the $X$-syndrome projector $\Pi_0$, and let
\begin{align}
 K_s=P_{\cC}C_s\Pi_sU_{\mathrm{nn}}P_{\cC}\ .
 \label{eq:kdef}
\end{align}
Since $C_s$ is of $Z$-type and has syndrome $s$, within the joint $+1$ eigenspace of the $Z$-stabilizers it maps the syndrome-$s$ subspace into $\cC$, so $C_s\Pi_sU_{\mathrm{nn}}\ket{\logicalPsi}=K_s\ket{\logicalPsi}$ for every code state~$\ket{\logicalPsi}\in\cC$.
Since
\begin{align}
[P_{\cC},\logicalZ]=[\Pi_s,\logicalZ]=[C_s,\logicalZ]=[U_{\mathrm{nn}},\logicalZ]=0\ ,
\end{align}
we conclude that $[K_s,\logicalZ]=0$. Because $K_s=P_{\cC}K_sP_{\cC}$ preserves the code space~$\cC$ (i.e., is a logical operator) and commutes with~$\logicalZ$, we conclude that
\begin{align}
 K_s=a_sP_{\cC}+b_s\logicalZ P_{\cC}
 \label{eq:kdecomposition}
\end{align}
for some scalars $a_s,b_s\in\mathbb C$.  These scalars are constrained by the following symmetry argument:  since the operator $\logicalX=X^{\otimes n}$ of Eq.~\eqref{eq:alldataX} commutes with $P_{\cC}$, $\Pi_s$ and $U_{\mathrm{nn}}$ (as this is generated by even-weight $Z$-type operators), and  because of Eq.~\eqref{eq:logicalXcsparity}, we have $\logicalX C_s\logicalX=(-1)^{t_s}C_s$. Hence $\logicalX K_s\logicalX=(-1)^{t_s}K_s$. Since $\logicalX\,\logicalZ\,\logicalX=-\logicalZ$, conjugating Eq.~\eqref{eq:kdecomposition} by $\logicalX$ gives
\[
 a_sP_{\cC}-b_s\logicalZ P_{\cC}=(-1)^{t_s}\left(a_sP_{\cC}+b_s\logicalZ P_{\cC}\right)\ .
\]
Thus $b_s=0$ if $t_s=0$ and $a_s=0$ if $t_s=1$. It follows that $K_s=\kappa_s\logicalZ^{\,t_s}P_{\cC}$ with $\kappa_s=a_s$ or $\kappa_s=b_s$, respectively, which gives the conditional logical operation in Lemma~\ref{lem:logicalaction}. Since $C_s$ is unitary, Eq.~\eqref{eq:pssample} gives $p_{\ket{\logicalPsi}}(s)=\|K_s\ket{\logicalPsi}\|^2=|\kappa_s|^2$ for every state $\ket{\logicalPsi}\in\cC$. For $p_{\ket{\logicalPsi}}(s)>0$, Eq.~\eqref{eq:finalstate} becomes $\ket{\psi_s}=(\kappa_s/|\kappa_s|)\logicalZ^{\,t_s}\ket{\logicalPsi}$, which proves Eq.~\eqref{eq:parity} and~\eqref{eq:rotatedlogicalrotation}.\hfill$\square$

\subsection{Syndrome distributions and IQP-circuit output probabilities}
\label{app:iqpamplitudes}
Here we prove Lemma~\ref{lem:syndromeiqp}, which relates syndrome probabilities to output probabilities of IQP circuits. Within this appendix we suppress the code and noise arguments of Eq.~\eqref{eq:parity} and write $p_{\ket{\logicalPsi}}(s)$, with $p_{\ket{\overline{\lambda}}}$ for logical computational-basis inputs $\ket{\overline{\lambda}}$, $\lambda\in\{0,1\}^k$. The lemma applies to all such basis states and their superpositions.

\noindent \emph{Proof of Lemma~\ref{lem:syndromeiqp}.}
The claimed identity follows by conjugating the noise and syndrome projectors by a CSS encoding map~$V$. To describe this in detail, partition an unencoded $n$-qubit register into $k$ logical qubits, $r_X$ syndrome qubits, and $r_Z=n-k-r_X$ qubits holding the eigenvalues of independent $Z$-stabilizers. Choose $X$-type logical representatives $\logicalX_j$ dual to the chosen $\logicalZ_j$, $j\in[k]$, with overall phase $+1$. There is a unitary encoding map~$V$ such that
\begin{align}
\begin{matrix}
 V^\dagger \logicalX_j V&=X_j\otimes I\otimes I\ ,\\
 V^\dagger \logicalZ_j V&=Z_j\otimes I\otimes I
 \qquad &\textrm{ for every }\qquad j\in [k]\\
 V^\dagger S_u^XV&=I\otimes X_u\otimes I\qquad &\textrm{ for every } u\in [r_X]\\
 V^\dagger S_v^ZV&=I\otimes I\otimes Z_v\qquad &\textrm{ for every }v\in [r_Z]\ .
 \end{matrix}
 \label{eq:cssencoding}
\end{align}
For the CSS generators and logical representatives used here, $V$ can be chosen to consist only of CNOT gates.

Consequently, for every physical-qubit support $\Gamma\subseteq[n]$ there are $\ell(\Gamma)\subseteq [k]$ and a subset $T(\Gamma)\subseteq[r_Z]$ such that
\begin{align}
 V^\dagger Z(\Gamma)V
 =Z(\ell(\Gamma))\otimes Z(\partial\Gamma)
       \otimes Z(T(\Gamma))\ .
 \label{eq:encodedpauli}
\end{align}
Here $\ell(\Gamma)$ is defined by 
the condition that $j\in \ell(\Gamma)$ for $j\in [k]$ if and only if 
$\logicalX_j Z(\Gamma)=(-1)Z(\Gamma)\logicalX_j$. The operator
$Z(\ell(\Gamma))$ 
 acts on the ``logical'' register~$(\mathbb{C}^2)^{\otimes k}$. 
 The next $r_X$ qubits are $X$-type syndrome registers, 
 with $Z(\partial \Gamma)$ corresponding to the syndrome bits that are flipped when applying the operator~$Z(\Gamma)$ to a code state. Finally, the last factor~$Z(T(\Gamma))$ in Eq.~\eqref{eq:encodedpauli} acts on an $r_Z$-qubit register corresponding to the $Z$-type stabilizers.

For $\lambda\in\{0,1\}^k$, fix the phases of the logical basis states by
\begin{align}
 \ket{\overline{\lambda}}
 =V\bigl(\ket\lambda\otimes\ket+^{\otimes r_X}
                       \otimes\ket0^{\otimes r_Z}\bigr)\ .
 \label{eq:encodedlogicalstate}
\end{align}
Let $\cI$ denote the elementary supports of the noise, and write $\theta_\Gamma$ for the  real angle associated with the operator~$Z(\Gamma)$, where $\Gamma\in \cI$. For Eq.~\eqref{eq:noise}, these are the singletons with $\theta_{\{q\}}=\varphi_q$ and the edges $e\in\cE$ with $\theta_e$ as in the main text. Define the noise Hamiltonian and unitary by
\begin{align}
 H^{\mathrm{noise}}=\sum_{\Gamma\in\cI}\theta_\Gamma Z(\Gamma)\qquad\textrm{ and }\qquad 
 \qquad U=e^{iH^{\mathrm{noise}}}\ .
 \label{eq:physicalhamiltonian}
\end{align}
Eq.~\eqref{eq:encodedpauli} implies
\begin{align}
 \begin{aligned}
 V^\dagger UV&=\exp\!\left(i\sum_{\Gamma\in\cI}\theta_\Gamma
 Z(\ell(\Gamma))\otimes Z(\partial\Gamma)
                    \otimes Z(T(\Gamma))\right)\ ,\\
 H_\lambda
 &=\sum_{\Gamma\in\cI}\theta_\Gamma(-1)^{\lambda\cdot\ell(\Gamma)}Z(\partial\Gamma)\ ,
 \qquad W_{\ket{\overline{\lambda}}}=e^{iH_\lambda}\ .
 \end{aligned}
 \label{eq:encodednoise}
\end{align}
Here we set  $\lambda\cdot \ell(\Gamma)=\sum_{j\in\ell(\Gamma)} \lambda_j$. Indeed, the logical factor acts on $\ket\lambda$ with eigenvalue $(-1)^{\lambda\cdot\ell(\Gamma)}$, while $Z(T(\Gamma))\ket0^{\otimes r_Z}=\ket0^{\otimes r_Z}$. Thus
\begin{align}
 V^\dagger U\ket{\overline{\lambda}}
 =\ket\lambda\otimes W_{\ket{\overline{\lambda}}}\ket+^{\otimes r_X}
                 \otimes\ket0^{\otimes r_Z}\ .
 \label{eq:encodednoisestate}
\end{align}
The syndrome projector transforms just as directly:
\begin{align}
 V^\dagger\Pi_sV
 =I\otimes\bigl(\HH^{\otimes r_X}\ket s\bra s
                   \HH^{\otimes r_X}\bigr)\otimes I\ .
 \label{eq:encodedsyndromeprojector}
\end{align}
For each syndrome $s$ define the normalized state
\begin{align}
 \ket{s;\lambda}
 =V\bigl(\ket\lambda\otimes\HH^{\otimes r_X}\ket s
                 \otimes\ket0^{\otimes r_Z}\bigr)\ .
 \label{eq:syndromestate}
\end{align}
These $2^{r_X}$ states form an orthonormal basis of the joint eigenspace with eigenvalue $+1$ for the $Z$-stabilizers and $(-1)^{\lambda_j}$ for each $\logicalZ_j$, with $\ket{0;\lambda}=\ket{\overline{\lambda}}$, and $\Pi_s\ket{t;\lambda}=\delta_{s,t}\ket{s;\lambda}$ by Eq.~\eqref{eq:encodedsyndromeprojector}. Unitarity of $V$ and Eq.~\eqref{eq:encodednoisestate} therefore give
\begin{align}
 \begin{aligned}
 A_\lambda(s)&=\bra{s;\lambda}U\ket{\overline{\lambda}}
 =\bra s\HH^{\otimes r_X}W_{\ket{\overline{\lambda}}}\ket+^{\otimes r_X}\ ,\\
 \Pi_sU\ket{\overline{\lambda}}&=A_\lambda(s)\ket{s;\lambda},\qquad
 \|\Pi_sU\ket{\overline{\lambda}}\|^2=|A_\lambda(s)|^2\ .
 \end{aligned}
 \label{eq:iqpbranches}
\end{align}
Since the terms of $H_\lambda$ commute,
\begin{align}
 W_{\ket{\overline{\lambda}}}=e^{iH_\lambda}
   =\prod_{\Gamma\in\cI}e^{i\theta_\Gamma(-1)^{\lambda\cdot\ell(\Gamma)}Z(\partial\Gamma)}\ .
 \label{eq:iqpgeneral}
\end{align}
Eqs.~\eqref{eq:iqpbranches} and~\eqref{eq:iqpgeneral} give Eqs.~\eqref{eq:syndromeiqpprobability} and~\eqref{eq:syndromeiqpcircuit}: the diagonal gates commute, the input is $\ket+^{\otimes r_X}$, and the syndrome is the outcome of measuring each register qubit in the $X$ basis. Each factor in $W_{\ket{\overline{\lambda}}}$ is supported on the bits marked by $\partial\Gamma$. The derivation applies to any set of $Z$-type elementary supports on a CSS code.

For a superposition $\ket{\logicalPsi}=\sum_\lambda a_\lambda\ket{\overline\lambda}$, Eq.~\eqref{eq:iqpbranches} gives $\Pi_sU\ket{\logicalPsi}=\sum_\lambda a_\lambda A_\lambda(s)\ket{s;\lambda}$. The states $\ket{s;\lambda}$ are orthonormal for distinct $\lambda$, so taking the squared norm proves Eq.~\eqref{eq:generalinputmixture} and its one-logical-qubit specialization, Eq.~\eqref{eq:inputmixture}.

Grouping the terms of the Hamiltonian by their syndrome gives
\begin{align}
 \begin{aligned}
 H_\lambda&=\chi_\lambda I+\sum_{\delta\neq0}
 \left[\sum_{\substack{\Gamma\in\cI\\\partial\Gamma=\delta}}
 \theta_\Gamma(-1)^{\lambda\cdot\ell(\Gamma)}\right]Z(\delta)\ ,\\
 \chi_\lambda&=\sum_{\substack{\Gamma\in\cI\\\partial\Gamma=0}}
 \theta_\Gamma(-1)^{\lambda\cdot\ell(\Gamma)}\ .
 \end{aligned}
 \label{eq:effhamiltonian}
\end{align}
Thus supports~$\Gamma$ for which $Z(\Gamma)$ causes 
an identical syndrome~$\delta$ contribute additively to one effective coupling, while trivial-syndrome terms contribute the (irrelevant) global phase $e^{i\chi_\lambda}$. For logical-zero input we write $H=H_{\ket{\overline{0}}}$, $W=W_{\ket{\overline{0}}}$ and $\chi=\chi_{\ket{\overline{0}}}$, as used below. This proves Lemma~\ref{lem:syndromeiqp}.\hfill$\square$

\subsection{Interaction graphs of IQP circuits}
\label{app:iqpgraphs}
\label{app:syndromeiqp}
\label{app:rotatedmapping}

\begin{figure*}[!t]
\centering
\begin{minipage}[c]{0.60\textwidth}
\centering
\subfloat[{Rotated patch, $d=5$.\label{fig:mappingd5-rotated}}]{\begin{minipage}{\linewidth}
 \centering
 \includegraphics[width=\linewidth]{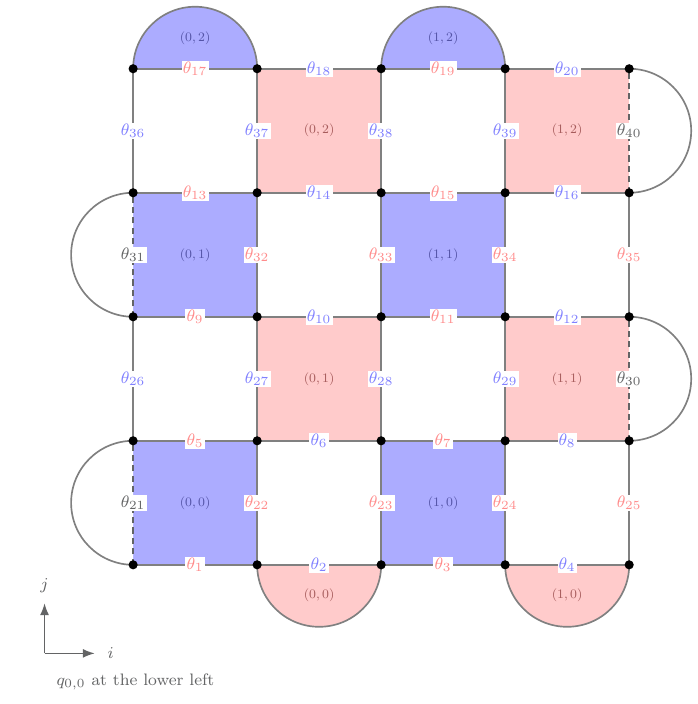}
\end{minipage}}
\end{minipage}\hfill
\begin{minipage}[c]{0.30\textwidth}
\centering
\begin{minipage}[t]{\linewidth}
\centering
\subfloat[{Auxiliary $A$, $D=3$.\label{fig:mappingd5-auxA}}]{\begin{minipage}{\linewidth}
 \centering
 \includegraphics[width=\linewidth]{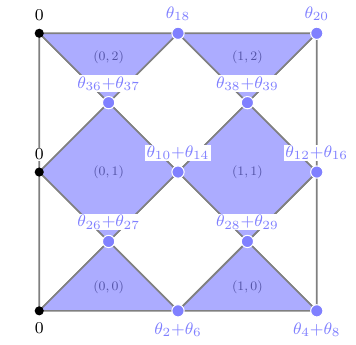}
\end{minipage}}
\end{minipage}
\par\vspace{1em}
\begin{minipage}[t]{\linewidth}
\centering
\subfloat[{Auxiliary $B$, $D=3$.\label{fig:mappingd5-auxB}}]{\begin{minipage}{\linewidth}
 \centering
 \includegraphics[width=\linewidth]{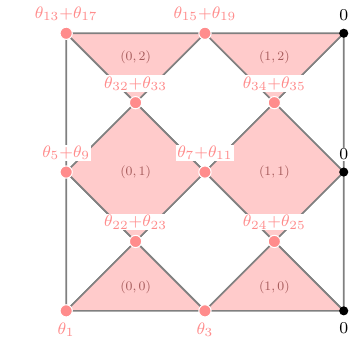}
\end{minipage}}
\end{minipage}
\end{minipage}
\caption{The assignment of Eq.~\eqref{eq:auxangles} in full for $d=5$ and $D=3$. (a)~The $2d(d-1)=40$ edge angles are numbered $\theta_{1+i+4j}=\theta^{\mathrm h}_{i,j}$ ($0\leq i\leq3$, $0\leq j\leq4$) and $\theta_{21+i+5j}=\theta^{\mathrm v}_{i,j}$ ($0\leq i\leq4$, $0\leq j\leq3$): horizontal bonds first, row by row from bottom to top, then vertical bonds ordered by the row of their lower endpoint. This numbering is unrelated to the six illustrative angles of Fig.~\ref{fig:mapping}. Every $X$ face carries the component-grid label $(u,v)$ of $S^{X,(C)}_{u,v}$, with blue identifying $C=A$ and red $C=B$; the corresponding auxiliary faces in (b),(c) carry the identical labels of $\widetilde S^{X,(C)}_{u,v}$. Colored physical qubits show the angle sums they receive; black sites marked $0$ carry angle zero. Each auxiliary code receives $18$ angles on $10$ of its $13$ qubits. The four dashed bonds $\{q_{0,0},q_{0,1}\}$, $\{q_{0,2},q_{0,3}\}$, $\{q_{4,1},q_{4,2}\}$ and $\{q_{4,3},q_{4,4}\}$ are weight-two $Z$-stabilizers and contribute only the phase $\chi=\theta_{21}+\theta_{30}+\theta_{31}+\theta_{40}$.}
\label{fig:mappingd5}
\end{figure*}

By Lemma~\ref{lem:syndromeiqp}, the syndrome distribution of a CSS code prepared in $\ket{\overline{0}}$ is the output distribution of the IQP circuit $W=\prod_{\Gamma\in\cI}e^{i\theta_\Gamma Z(\partial\Gamma)}$ on the $r_X$-qubit syndrome register. If every support changes the eigenvalue of at most two $X$-stabilizers, Eq.~\eqref{eq:effhamiltonian} gives
\begin{align}
 \begin{aligned}
 W&=e^{i\chi}\prod_{\{\xi,\eta\}\in E}e^{iJ_{\{\xi,\eta\}}Z_\xi Z_\eta}
 \prod_{\xi\in V'}e^{ih_\xi Z_\xi}\ ,\\
 J_{\{\xi,\eta\}}&=\sum_{\partial\Gamma=\{\xi,\eta\}}\theta_\Gamma\ ,\quad
 h_\xi=\sum_{\partial\Gamma=\{\xi\}}\theta_\Gamma\ ,\quad
 \chi=\sum_{\partial\Gamma=\emptyset}\theta_\Gamma\ ,
 \end{aligned}
 \label{eq:syndromeiqpising}
\end{align}
where $Z_\xi$ is Pauli $Z$ on the syndrome qubit~$\xi$, the sums run over $\Gamma\in\cI$, the edge set $E$ consists of the two-element syndromes and the field-vertex set $V'\subseteq V$ of the one-element syndromes $\partial\Gamma$, $\Gamma\in\cI$. We call $\mathcal G=(V,E)$, with $V$ the set of $X$-stabilizers, the interaction graph. Since each support has exactly one syndrome, the sums in Eq.~\eqref{eq:syndromeiqpising} run over pairwise disjoint subsets of~$\cI$, so arbitrary couplings $J$ and fields $h$ on $\mathcal G$ are realized by giving one support with the corresponding syndrome the desired angle and setting all other angles to zero; this is used in Appendix~\ref{app:hardness}. It remains to list the syndromes of single qubits and of edges, which we do in coordinates.

For the rotated surface code of odd distance $d=2D-1$, let $\xi^{(C)}_{u,v}$, $(u,v)\in G$, denote the syndrome qubit of the stabilizer $S^{X,(C)}_{u,v}$ of Eqs.~\eqref{eq:bulkxchecks} and~\eqref{eq:boundaryxchecks}, and let $V^{(C)}=\{\xi^{(C)}_{u,v}:(u,v)\in G\}$. Geometrically, the face of $S^{X,(C)}_{u,v}$ in Fig.~\ref{fig:surface_codes}(a) has lower-left corner
\begin{align}
 \rho(\xi^{(A)}_{u,v})=(2u,2v)\ ,\qquad
 \rho(\xi^{(B)}_{u,v})=(2u+1,2v-1)\ ,
 \label{eq:geocomponents}
\end{align}
which give the lower-left corner $(x,y)$ of the face of the stabilizer in Fig.~\ref{fig:surface_codes}(a): the face contains the qubits $q_{i,j}$ with $i\in\{x,x+1\}$ and $j\in\{y,y+1\}$, and the faces on the bottom and top boundaries have $y=-1$ and $y=d-1$, respectively. The map $\rho$ is a bijection from $V^{(A)}\sqcup V^{(B)}$ onto
\begin{align}
 \begin{aligned}
 \cX_d^{\mathrm{geo}}=\{(x,y)\in\mathbb Z^2:\ &0\leq x\leq d-2\ ,\\
 &-1\leq y\leq d-1,\ x+y\text{ even}\}\ ,
 \end{aligned}
 \label{eq:geofaces}
\end{align}
with $V^{(A)}$ the faces of even $x$ and $V^{(B)}$ those of odd~$x$. For the standard (unrotated) surface code of distance $D$, let $\xi_{u,v}$, $0\leq u\leq D-2$, $0\leq v\leq D-1$, denote the syndrome qubit of the $X$-stabilizer with grid label $(u,v)$, as in Fig.~\ref{fig:mapping}(b),(c). The $m=D^2+(D-1)^2$ physical qubits are $h_{u,v}$, $0\leq u,v\leq D-1$, lying between the faces $(u-1,v)$ and $(u,v)$, and $w_{u,v}$, $0\leq u,v\leq D-2$, lying between the faces $(u,v)$ and $(u,v+1)$, so that
\begin{align}
 \partial\{h_{u,v}\}=\{\xi_{u-1,v},\xi_{u,v}\}\ ,\qquad
 \partial\{w_{u,v}\}=\{\xi_{u,v},\xi_{u,v+1}\}\ ,
 \label{eq:auxstandardsyndromes}
\end{align}
where the undefined symbols $\xi_{-1,v}$ and $\xi_{D-1,v}$ are to be omitted: the qubits $h_{0,v}$ and $h_{D-1,v}$ form the two rough boundaries and flip a single stabilizer each. For $D\geq3$ these $m$ syndromes are distinct and nonempty, and they exhaust the edges of the $(D-1)\times D$ grid graph on $\{\xi_{u,v}\}$ and the vertices of its two boundary columns $u=0$ and $u=D-2$.

Figure~\ref{fig:mappingd5} illustrates the complete angle assignment of Eq.~\eqref{eq:auxangles} for $d=5$ and $D=3$.

\begin{appendixlemma}[Surface-code interaction graphs]
\label{lem:surfacecodeiqp}
Let $d\geq5$ be odd and $D=(d+1)/2$.
\begin{enumerate}[(a)]
\item On the rotated surface code $\Crotated{d}$, in the coordinates of Eq.~\eqref{eq:geofaces},
\begin{align}
 \partial\{q_{i,j}\}&=\cX_d^{\mathrm{geo}}\cap\bigl(\{i-1,i\}\times\{j-1,j\}\bigr)\ ,
 \label{eq:sqgeosyndrome}\\
 \partial\{q_{i,j},q_{i+1,j}\}&=\cX_d^{\mathrm{geo}}\cap\bigl(\{i-1,i+1\}\times\{j-1,j\}\bigr)\ ,
 \label{eq:hpairgeo}\\
 \partial\{q_{i,j},q_{i,j+1}\}&=\cX_d^{\mathrm{geo}}\cap\bigl(\{i-1,i\}\times\{j-1,j+1\}\bigr)\ .
 \label{eq:vpairgeo}
\end{align}
Thus a single qubit flips two faces at displacement $(1,\pm1)$, or one face in the left column $x=0$ or right column $x=d-2$ if $i=0$ or $i=d-1$; an edge flips two faces at displacement $(2,0)$ or $(0,2)$, one face in the column $x=1$ or $x=d-3$ if it is horizontal with $i=0$ or $i=d-2$, or no face if its operator $Z(e)$ is one of the $d-1$ weight-two $Z$-stabilizers. Every such set of faces is the syndrome of some qubit or edge.
\item The interaction graphs of Eq.~\eqref{eq:syndromeiqpising} are as follows.
\begin{enumerate}[(i)]
\item Crosstalk noise $U_{\mathrm{nn}}(\bm\theta)$ on $\Crotated{d}$: $\mathcal G=\mathcal G^{(A)}\sqcup\mathcal G^{(B)}$, where $\mathcal G^{(C)}$ is the $(D-1)\times D$ grid graph on $V^{(C)}$, with edges $\{\xi^{(C)}_{u,v},\xi^{(C)}_{u',v'}\}$ for $|u-u'|+|v-v'|=1$, and the field vertices are the right column $u=D-2$ of $\mathcal G^{(A)}$ and the left column $u=0$ of $\mathcal G^{(B)}$ (Fig.~\ref{fig:pairiqp}(a)); $\chi$ is the sum of the angles of the $d-1$ weight-two $Z$-stabilizers, and
\begin{align}
 W=e^{i\chi}\,W^{(A)}\otimes W^{(B)}\ ,
 \label{eq:pairising}
\end{align}
with $W^{(C)}$ the product of the factors supported on $V^{(C)}$.
\item Single-qubit noise $U_{\mathrm{sq}}(\bm\varphi)$ on $\Cstandard{D}$: $\mathcal G$ is the $(D-1)\times D$ grid graph on $\{\xi_{u,v}\}$, with field vertices on both columns $u=0$ and $u=D-2$, and $\chi=0$ (Fig.~\ref{fig:pairiqp}(b)). The map $\xi^{(C)}_{u,v}\mapsto\xi_{u,v}$ is an isomorphism of $\mathcal G^{(C)}$ in case~(i) onto $\mathcal G$ which maps the field vertices of $\mathcal G^{(A)}$ onto the right column $u=D-2$ and those of $\mathcal G^{(B)}$ onto the left column $u=0$.
\item Combined noise $U_{\mathrm{sq}}(\bm\varphi)U_{\mathrm{nn}}(\bm\theta)$ on $\Crotated{d}$: identifying $V$ with $\cX_d^{\mathrm{geo}}$ via $\rho$,
\begin{align}
 \begin{aligned}
 E&=E_{\mathrm{sq}}\cup E_{\mathrm{nn}}\ ,\\
 E_{\mathrm{sq}}&=\bigl\{\{\xi,\eta\}:\xi-\eta\in\{\pm(1,1),\pm(1,-1)\}\bigr\}\ ,\\
 E_{\mathrm{nn}}&=\bigl\{\{\xi,\eta\}:\xi-\eta\in\{\pm(2,0),\pm(0,2)\}\bigr\}\ ,\\
 V'&=\{(x,y)\in\cX_d^{\mathrm{geo}}:x\in\{0,1,d-3,d-2\}\}\ ,
 \end{aligned}
 \label{eq:combinedising}
\end{align}
with $\chi$ as in case~(i). In the coordinates $(a,b)=((x+y)/2,(x-y)/2)$, the edges of $E_{\mathrm{sq}}$ are axial and those of $E_{\mathrm{nn}}$ are diagonal: $\mathcal G$ is the subgraph of the king lattice induced on $\cX_d^{\mathrm{geo}}$ (Fig.~\ref{fig:kingiqp}). The edges of $E_{\mathrm{nn}}$ lie within $V^{(A)}$ or within $V^{(B)}$, whereas every edge of $E_{\mathrm{sq}}$ joins $V^{(A)}$ to $V^{(B)}$.
\end{enumerate}
\item Writing $\theta^{\mathrm h}_{i,j}=\theta_{\{q_{i,j},q_{i+1,j}\}}$ and $\theta^{\mathrm v}_{i,j}=\theta_{\{q_{i,j},q_{i,j+1}\}}$, the single-qubit angles assigned to the auxiliary codes by Eq.~\eqref{eq:anglesum} are
\begin{align}
 \begin{aligned}
 \varphi^{(A)}_{h_{u,v}}&=\theta^{\mathrm h}_{2u-1,2v}+\theta^{\mathrm h}_{2u-1,2v+1}\ ,
 &&1\leq u\leq D-1\ ,\\
 \varphi^{(A)}_{w_{u,v}}&=\theta^{\mathrm v}_{2u,2v+1}+\theta^{\mathrm v}_{2u+1,2v+1}\ ,\\
 \varphi^{(B)}_{h_{u,v}}&=\theta^{\mathrm h}_{2u,2v-1}+\theta^{\mathrm h}_{2u,2v}\ ,
 &&0\leq u\leq D-2\ ,\\
 \varphi^{(B)}_{w_{u,v}}&=\theta^{\mathrm v}_{2u+1,2v}+\theta^{\mathrm v}_{2u+2,2v}\ ,\\
 \varphi^{(A)}_{h_{0,v}}&=\varphi^{(B)}_{h_{D-1,v}}=0\ ,
 \end{aligned}
 \label{eq:auxangles}
\end{align}
with $0\leq u,v\leq D-2$ in the second and fourth lines and $0\leq v\leq D-1$ otherwise, where a term with a qubit index outside $\{0,\ldots,d-1\}$ is omitted.
\end{enumerate}
\end{appendixlemma}

\emph{Proof.} (a) The face $(x,y)$ contains $q_{i,j}$ if and only if $x\in\{i-1,i\}$ and $y\in\{j-1,j\}$, and $Z(\Gamma)$ flips a face if and only if the face contains an odd number of qubits of~$\Gamma$; this gives Eqs.~\eqref{eq:sqgeosyndrome}--\eqref{eq:vpairgeo}. Of the four candidate faces, the parity condition of Eq.~\eqref{eq:geofaces} selects the two at the stated displacement, and the condition $0\leq x\leq d-2$ removes those with $x\in\{-1,d-1\}$. For a vertical edge both selected faces have the same $x$, so either both or none are removed; the syndrome is empty for $\{q_{0,j},q_{0,j+1}\}$ with $j$ even and $\{q_{d-1,j},q_{d-1,j+1}\}$ with $j$ odd, the weight-two $Z$-stabilizers. Conversely, two faces at displacement $(1,\pm1)$ share a qubit, two faces at displacement $(2,0)$ or $(0,2)$ are separated by an edge, and a face in the column $x=0$, $1$, $d-3$ or $d-2$ contains a qubit $q_{0,j}$, $q_{1,j}$, $q_{d-2,j}$ or $q_{d-1,j}$, respectively.

(b) By Eq.~\eqref{eq:syndromeiqpising}, $E$ and $V'$ are the sets of two- and one-element syndromes, which are listed in (a) and in Eq.~\eqref{eq:auxstandardsyndromes}. In case~(i), the displacements $(2,0)$ and $(0,2)$ preserve the parity of $x$, hence the component, and by Eq.~\eqref{eq:geocomponents} correspond to $(u,v)\mapsto(u\pm1,v)$ and $(u,v)\mapsto(u,v\pm1)$; the columns $x=d-3$ and $x=1$ consist of the faces $\rho(\xi^{(A)}_{D-2,v})$ and $\rho(\xi^{(B)}_{0,v})$. Since no factor of Eq.~\eqref{eq:syndromeiqpising} couples $V^{(A)}$ to $V^{(B)}$, Eq.~\eqref{eq:pairising} follows. In case~(iii), the map $(x,y)\mapsto(a,b)$ is injective on $\cX_d^{\mathrm{geo}}$ and sends the displacements $(1,1),(1,-1),(2,0),(0,2)$ to $(1,0),(0,1),(1,1),(1,-1)$; the displacements $(1,\pm1)$ change the parity of~$x$.

(c) By Eq.~\eqref{eq:pairgroups}, the edges mapped to a qubit $q$ of the auxiliary code $C$ are those whose syndrome is $\partial\{q\}$ of Eq.~\eqref{eq:auxstandardsyndromes}, read in component~$C$. Substituting Eq.~\eqref{eq:geocomponents} into Eqs.~\eqref{eq:hpairgeo} and~\eqref{eq:vpairgeo} lists these edges. The singletons $\{\xi^{(A)}_{0,v}\}$ and $\{\xi^{(B)}_{D-2,v}\}$ are not syndromes of edges, so $h^{(A)}_{0,v}$ and $h^{(B)}_{D-1,v}$ receive angle zero.
\hfill$\square$

Appendix~\ref{app:hardness} embeds a rectangular king-lattice layout into case~(iii) by setting the angles of all supports outside the layout to zero. If $W=e^{iK_Q}\otimes I_{Q^{\mathrm c}}$ acts trivially on a set $Q^{\mathrm c}$ of syndrome qubits, then
\begin{align}
 \HH^{\otimes r_X}\,W\ket+^{\otimes r_X}
 =\bigl(\HH^{\otimes|Q|}e^{iK_Q}\ket+^{\otimes|Q|}\bigr)
 \otimes\ket0^{\otimes|Q^{\mathrm c}|}\ ,
 \label{eq:iqpinactivesites}
\end{align}
so the outcomes on $Q^{\mathrm c}$ are deterministically zero and the distribution on $Q$ is that of the smaller register.

\begin{figure*}[!tbp]
\centering
\begin{minipage}[t]{184pt}
\centering
\subfloat[{Rotated surface code, $d=7$, with nearest-neighbor $ZZ$ crosstalk.\label{fig:pairiqp-rotated}}]{\begin{minipage}{\linewidth}
 \centering
 \includegraphics{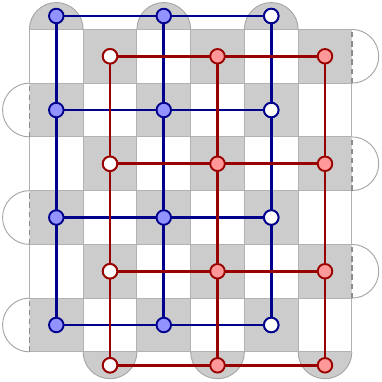}
\end{minipage}}
\end{minipage}\hfill
\begin{minipage}[t]{150pt}
\centering
\subfloat[{Auxiliary standard (unrotated) surface code $A$, $D=4$, with
          single-qubit $Z$ rotations.\label{fig:pairiqp-auxA}}]{\begin{minipage}{\linewidth}
 \centering
 \includegraphics{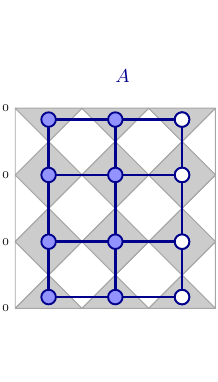}
\end{minipage}}
\end{minipage}\hfill
\begin{minipage}[t]{150pt}
\centering
\subfloat[{Auxiliary standard (unrotated) surface code $B$, $D=4$, with
          single-qubit $Z$ rotations.\label{fig:pairiqp-auxB}}]{\begin{minipage}{\linewidth}
 \centering
 \includegraphics{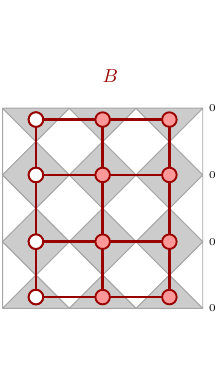}
\end{minipage}}
\end{minipage}
\caption{The same IQP circuit from two noise models. Nearest-neighbor $ZZ$
crosstalk on the rotated code of distance $d=7$ in (a) induces two
disconnected syndrome lattices, blue for $A$ and red for $B$; single-qubit
$Z$ rotations on the two auxiliary standard (unrotated) surface codes $A$ and
$B$ of distance $D=4$ in (b) and (c) induce the same lattices. Circles
represent $X$-stabilizer syndrome qubits, outlined in blue for $A$ and red for
$B$; a solid fill in the component color means no one-body field, and a white
fill marks a one-body field. A thick colored line joins the two circles of a
two-body coupling. Grey faces are $X$-stabilizers and white faces are $Z$-stabilizers. Grey dashed bonds in (a) are the physical edges with trivial syndrome,
which contribute only the scalar $e^{i\chi}$, and crossings without circles
are not vertices. Each component has $12$ vertices, $17$ edges, and $4$ field
vertices. Corresponding vertices carry the same component $C$ and grid index
$(u,v)$, with indices omitted for clarity. Assigning the auxiliary angles by
Eq.~\eqref{eq:auxangles} matches every edge coupling and field. The opposite auxiliary rough
boundaries carry zero rotation angles, marked $0$.}
\label{fig:pairiqp}
\end{figure*}

\subsection{Proof of Lemma~\ref{lem:mapping}}
\label{app:mappingproof}

Figure~\ref{fig:pairiqp} compares the syndrome-register IQP circuits: crosstalk on the rotated surface code in (a) and single-qubit noise on the two auxiliary standard (unrotated) surface codes in (b) and (c). We show that their Hamiltonians agree under the stabilizer correspondence of Sec.~\ref{sec:reduction}, up to the scalar from the edges with trivial syndrome. We first work at logical-zero inputs. For crosstalk, Eq.~\eqref{eq:iqpgeneral} becomes
\begin{align}
 \begin{aligned}
 A(s)&=\bra s\HH^{\otimes r_X}W_{\mathrm{nn}}\ket+^{\otimes r_X}\ ,\\
 H_{\mathrm{nn}}&=\sum_{e\in\cE}\theta_eZ(\partial e)\ ,
 \qquad W_{\mathrm{nn}}=e^{iH_{\mathrm{nn}}}\ .
 \end{aligned}
 \label{eq:nniqp}
\end{align}
In Fig.~\ref{fig:pairiqp}(a), the blue and red lattices are disconnected. This is the decomposition of Eq.~\eqref{eq:pairising}: by Lemma~\ref{lem:surfacecodeiqp}(b)(i), every nonzero $\partial e$ is an edge or a field vertex of exactly one component graph $\mathcal G^{(C)}$, $C\in\{A,B\}$. Under the stabilizer correspondence of Sec.~\ref{sec:reduction}, its restriction $(\partial e)^{(C)}$ equals $\partial^{(C)}\{q\}$ for a physical qubit $q$ of the auxiliary code. For $D\geq3$, the single-qubit syndromes in Eq.~\eqref{eq:auxstandardsyndromes} are distinct and nonzero, so $q$ is unique. For each of the remaining $d-1$ edges, $Z(e)$ is one of the weight-two $Z$-stabilizers on the left and right boundaries and has trivial syndrome.

Order the syndrome register by $s=(s^{(A)},s^{(B)})$. Each block contains $|G|=D(D-1)=r_X/2$ qubits. Grouping physical edges with the same auxiliary image now gives the Hamiltonian identity
\begin{align}
 \begin{aligned}
 H_{\mathrm{nn}}
 &=\chi I+H^{(A)}\otimes I^{(B)}
              +I^{(A)}\otimes H^{(B)}\ ,\\
 H^{(C)}
 &=\sum_q\left[\sum_{e:\ \iota(e)=(C,q)}\theta_e\right]
                   Z^{(C)}(\partial^{(C)}\{q\})\\
 &=\sum_q\varphi_q^{(C)}Z^{(C)}(\partial^{(C)}\{q\})\ ,
 \qquad
 \chi=\sum_{e:\partial e=0}\theta_e\ .
 \end{aligned}
 \label{eq:hefffactor}
\end{align}
Here $Z^{(C)}$ acts on the syndrome register of component $C$. In Fig.~\ref{fig:pairiqp}(b), an interior auxiliary qubit connects the two syndrome vertices corresponding to the $X$-stabilizers it touches, while an active rough-boundary qubit gives a field on one syndrome vertex. Eq.~\eqref{eq:auxstandardsyndromes} gives exactly the edges and field vertices of (a), under the stabilizer correspondence. The zero angles on the opposite rough boundaries remove the unused fields. Applying Eq.~\eqref{eq:encodednoise} to $U^{(C)}$ at input $\ket{\overline{0}}^{(C)}$ therefore gives $H^{(C)}$: Eq.~\eqref{eq:anglesum} matches the coefficient of each term as well as its support. Thus the two noise models realize the same weighted IQP circuit, up to $e^{i\chi}$.

Exponentiating Eq.~\eqref{eq:hefffactor} yields
\begin{align}
 \begin{aligned}
 W_{\mathrm{nn}}&=e^{i\chi}(W^{(A)}\otimes W^{(B)})\ ,\\
 W^{(C)}&=e^{iH^{(C)}}
     =\prod_qe^{i\varphi_q^{(C)}Z^{(C)}(\partial^{(C)}\{q\})}\ .
 \end{aligned}
 \label{eq:wfactor}
\end{align}
Applying Lemma~\ref{lem:syndromeiqp} and the amplitude identity~\eqref{eq:iqpbranches} to auxiliary code $C$ at input $\ket{\overline{0}}^{(C)}$ gives, with $p_{\ket{\overline{0}}}^{(C)}$ abbreviated as $p^{(C)}$,
\begin{align}
 \begin{aligned}
 A^{(C)}(s^{(C)})&=\bra{s^{(C)}}\HH^{\otimes r_X/2}W^{(C)}\ket+^{\otimes r_X/2}\ ,\\
 p^{(C)}(s^{(C)})&=|A^{(C)}(s^{(C)})|^2\ .
 \end{aligned}
 \label{eq:auxiqp}
\end{align}
The input $\ket+^{\otimes r_X}$ and the $X$-basis measurement factor over the two registers, so
\begin{align}
 A(s^{(A)},s^{(B)})=e^{i\chi}A^{(A)}(s^{(A)})A^{(B)}(s^{(B)})\ .
 \label{eq:ampfactor}
\end{align}
Taking the squared modulus proves Eq.~\eqref{eq:factor} at logical-zero inputs.

State independence of the left-hand side of Eq.~\eqref{eq:factor} follows from Lemma~\ref{lem:logicalaction}. For the right-hand side, let $\mathcal B_A$ and $\mathcal B_B$ denote the left rough boundary of $\Caux{A}$ and the right rough boundary of $\Caux{B}$, respectively. The assigned angles $\varphi_q^{(C)}$ vanish on $\mathcal B_C$, so we choose the logical representatives
\begin{equation}
 \overline X^{(C)}=\prod_{q\in\mathcal B_C}X_q^{(C)},
 \quad C\in\{A,B\}.
 \label{eq:auxlogicalx}
\end{equation}
Each commutes with every active single-qubit rotation, giving $\ell(\Gamma)=0$ for all active supports and hence $W_{\ket{\overline{1}}}=W_{\ket{\overline{0}}}$ in Lemma~\ref{lem:syndromeiqp}. Both factors on the right-hand side are therefore state independent, for either parity of $D$.

Independent auxiliary samples consequently produce exactly the required syndrome of the original code for every encoded input. This proves Lemma~\ref{lem:mapping}.

\begin{figure}[!t]
\centering
\includegraphics{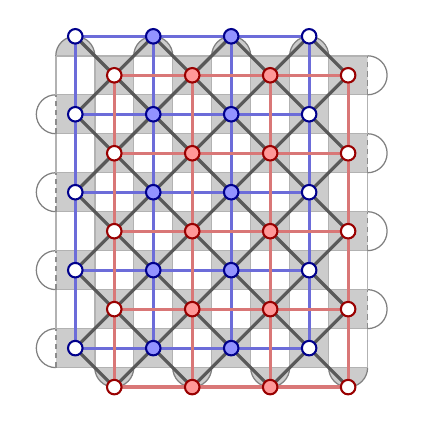}
\caption{Combined-noise interaction graph at $d=9$, with the same color and marker conventions as Fig.~\ref{fig:pairiqp}. Thin dark edges come from single-qubit rotations and colored edges from $ZZ$-rotations. The graph has $40$ vertices, $63$ single-qubit edges, $62$ $ZZ$-edges, and $20$ field vertices; the eight dashed bonds contribute only a scalar phase. The single-qubit edges join $A$ and $B$. In the geometric face coordinates $(x,y)$, the edges have displacements $(1,\pm1)$, $(2,0)$, or $(0,2)$, up to sign; the graph is isomorphic to a finite subgraph of the king lattice.}
\label{fig:kingiqp}
\end{figure}

Figure~\ref{fig:kingiqp} shows why the same factorization does not apply to combined noise: the single-qubit terms join the two components of Fig.~\ref{fig:pairiqp}. Their union gives the king-lattice interaction graph of Eq.~\eqref{eq:combinedising}, used in Sec.~\ref{sec:hardness}.

\section{Numerical methods and threshold extrapolation}
\label{app:thresholdfit}

\emph{Sampling and decoding.}
For every distance, angle and model of Sec.~\ref{sec:numerics}, we use $N=200\,000$ independent samples, with separate random streams for the two auxiliary codes and for the physical twirl. Coherent samples use logical-zero auxiliary inputs, the angle sums of Eq.~\eqref{eq:anglesum}, and sequential FLO sampling of physical $X$ outcomes, whose products give the auxiliary stabilizer syndromes. For the twirl, each edge of the original code is selected independently with probability $\sin^2\theta$; repeated $Z$ operators on a qubit cancel before its syndrome is computed.

Both models use the same MWPM decoder on the single-qubit $Z$-error graph. Each physical qubit contributes an edge between its adjacent $X$ checks, or an edge to a rough boundary if it touches only one check. All physical edges have unit weight, including boundary edges; matching costs are shortest-path lengths. Defects are paired with each other or with the nearest rough boundary to minimize the total cost. Equal-cost choices are resolved deterministically using a fixed ordering of checks and qubits, identical for both models. The correction is the symmetric difference of the selected paths, and its syndrome is checked against the sampled syndrome before its weight parity is recorded. No even-weight constraint or pair-noise reweighting is imposed.

\emph{Numerical precision.}
FLO covariance updates use binary64 arithmetic. Probability bounds and normalization identities are checked with absolute tolerances $2\times10^{-10}$ and $2\times10^{-9}$, respectively. A failed check or a covariance update requiring division by a conditional probability below $10^{-14}$ triggers a repeat of the entire auxiliary sample with 113-bit precision and the same random draws. The corresponding tolerances are $10^{-24}$, $10^{-22}$ and $10^{-28}$. Small negative weights within tolerance are set to zero. Validation includes exact small-instance probabilities and comparisons using the same random draws at both precisions through $d=37$. These checks address numerical stability; the intervals below quantify Monte Carlo sampling uncertainty.

\emph{Sampling uncertainty.}
We estimate $f$ and $P^L$ as in Eq.~\eqref{eq:estimator}. The binomial standard error of $\widehat P^L$ is $2\sqrt{\hat f(1-\hat f)/N}$. Plotted intervals are pointwise 95\% Wilson intervals for $f$, doubled for $P^L$; the display cutoff in Fig.~\ref{fig:numerics} does not remove observations from the fits. These intervals quantify sampling uncertainty at each point, rather than uncertainty in the extrapolated threshold.

\emph{Joint scaling fit.}
For coherent noise we fit $0.032\leq t\leq0.038$; for the physical twirl we fit $0.047\leq t\leq0.053$, where $t=\theta/\pi$. Each window contains seven angles at each of the four distances, hence $28$ observations. We fit Eq.~\eqref{eq:thresholdscaling} by minimizing
\begin{align}
 \chi^2&=\sum_i\frac{[\hat f_i-A-Bx_i-Cx_i^2]^2}{\sigma_i^2},\nonumber\\
 \sigma_i^2&=\frac{q_i(1-q_i)}{N_i},\qquad
 q_i=\frac{M_i+1/2}{N_i+1},
 \label{eq:fitweights}
\end{align}
where $i$ labels distance--angle points, $M_i$ is the observed failure count among $N_i$ samples, and $x_i=(t_i-t_c)d_i^{1/\nu}$. The smoothed binomial variance keeps the weights finite even for zero observed failures. For each trial $(t_c,\nu)$, the unrestricted coefficients $A,B,C$ are obtained by linear weighted least squares. The remaining minimization bounds $t_c$ to its fit window and $0.2\leq\nu\leq10$, and uses $20$ starting points: five window fractions $0.15,0.35,0.50,0.65,0.85$, each with $\nu=0.5,0.8,1.5,3$. We retain the converged fit with the smallest $\chi^2$. No correction-to-scaling term is included; $\chi^2_{\mathrm{red}}=\chi^2/(28-5)$. Figure~\ref{fig:thresholdfit} shows the joint fits reported in Table~\ref{tab:threshold}.

\emph{Distance and window sensitivity.}
We repeat the complete five-parameter fit four times, omitting one distance at a time. If $t_c^{(-j)}$ denotes the result with distance $j$ omitted, the quoted spread is
\begin{align}
 s_{\mathrm{JK}}&=\left[\frac{m-1}{m}\sum_{j=1}^m
 (t_c^{(-j)}-\bar t_c)^2\right]^{1/2},\nonumber\\
 \bar t_c&=\frac1m\sum_{j=1}^m t_c^{(-j)},\qquad m=4.
 \label{eq:jackknife}
\end{align}
Omitting $d=13,21,29,37$ in that order gives $t_c=0.034525,0.034123,0.034198,0.033891$ for coherent noise and $0.049825,0.049265,0.049391,0.048964$ for the twirl. Thus $s_{\mathrm{JK}}=0.000394$ and $0.000536$, respectively. Narrowing both window endpoints inward by $0.001$ gives $t_c=0.034213$ and $0.049320$; widening them outward by $0.001$ gives $0.034146$ and $0.049341$. The ordering is unchanged in these checks. The jackknife spread is a measure of distance sensitivity, not a confidence interval accounting for scaling-model bias, which is indicated by the large $\chi^2_{\mathrm{red}}$ values.

\begin{figure*}[!tbp]
\centering
\includegraphics[width=\textwidth]{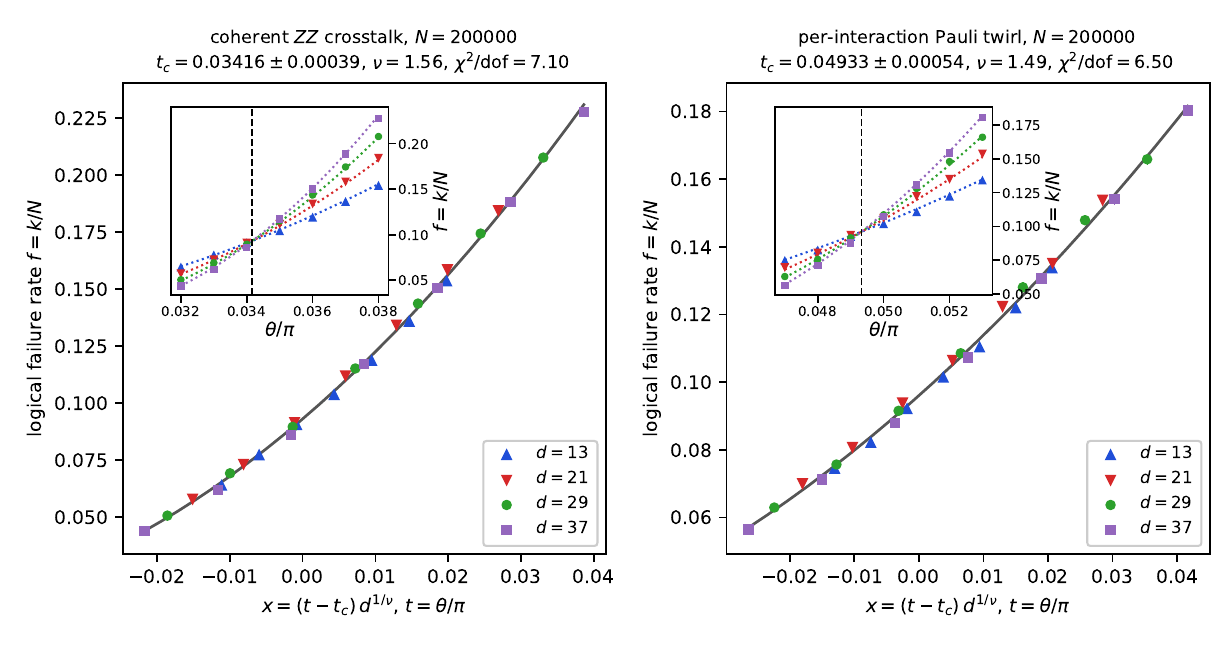}
\caption{Finite-size-scaling collapse for coherent noise (left) and
its physical twirl (right). Main panels: the measured failure rate
$\hat f=M/N$ for $d=13,21,29,37$ against the scaling variable
$x=(t-t_c)d^{1/\nu}$, with the fitted parabola $A+Bx+Cx^2$ of
Eq.~\eqref{eq:thresholdscaling} (gray). Insets: the same data before
rescaling, against $t=\theta/\pi$, over the fitted angle windows,
with the dashed line at the extrapolated $t_c$. Each point uses $200\,000$
samples; error bars show pointwise 95\% Wilson intervals for $f$
and are generally smaller than the markers. Panel titles quote $t_c$ with the jackknife spread
$s_{\mathrm{JK}}$ of Eq.~\eqref{eq:jackknife}, together with $\nu$ and
$\chi^2_{\mathrm{red}}$.}
\label{fig:thresholdfit}
\end{figure*}

\section{Postselected universality on the king lattice}
\label{app:hardness}

\subsection{The sampling problem}
The king lattice has vertex set $\mathbb Z^2$ and edges between
distinct sites $u,v$ with $\|u-v\|_\infty=1$. Our conventions are
\[
 \begin{aligned}
  R_j(\alpha)&=e^{i\alpha Z_j}, &
  G_{jk}(\theta)&=e^{i\theta Z_jZ_k}\ ,\\
  \ket{x}_X&=\bigotimes_j
  \frac{\ket0+(-1)^{x_j}\ket1}{\sqrt2}\ .
 \end{aligned}
\]
Thus outcome $0$ means $\ket+$ and outcome $1$ means $\ket-$.
For a finite rectangular region $\Lambda$, define
\begin{align}\label{ghz:eq:model}
 \begin{aligned}
  \mathcal H_J&=\sum_{\{u,v\}\in E(\Lambda)}J_{uv}Z_uZ_v\ ,\\
  U_J&=e^{i\mathcal H_J}\ ,\\
  p_J(x)&=\left|\bra{x}_XU_J
  \ket+^{\otimes |\Lambda|}\right|^2\ .
 \end{aligned}
\end{align}
There are no one-body fields. Couplings are
individually specified; unused king-lattice edges have coupling zero.
Let $\mathcal K$ be the family in Eq.~\eqref{ghz:eq:model} with
\[
 J_{uv}\in\{k\pi/8:k=0,1,\ldots,15\}\ ,
\]
and at most four nonzero couplings at each site.

\subsection{Replace Hadamards and expand controlled-\texorpdfstring{$Z$}{Z} gates}
\label{ghz:sec:iqp}
Start with a nearest-neighbor brickwork circuit $C$ on $q$ wires,
over $H$, $R(\pi/8)$ and $\CZ$, with $\ket+$ inputs and $X$
measurements. Its two-qubit layers consist of disjoint CZ gates
on neighboring wires, with the pairings staggered between layers.
Single-qubit layers can be separated into diagonal rotations and
Hadamards on subsets of wires. This form is universal with polynomial
overhead; boundary Hadamards give the stated input and measurement
conventions.

For a regular layout, append $HH=I$ on every wire after each of
the $T$ layers. After a diagonal layer each wire then has two
Hadamards. In a Hadamard layer it has three if that layer acts on
it, and two otherwise. This padding does not change $C$.
As in the post-IQP construction of Bremner, Jozsa and
Shepherd~\cite{BremnerJozsaShepherd2011}, replace every Hadamard by
\begin{align}\label{ghz:eq:hadamard}
 (\bra+_a\otimes I_b)\CZ_{ab}
 (\ket\psi_a\otimes\ket+_b)
 =2^{-1/2}H\ket\psi_b\ .
\end{align}
The identity also holds when $a$ is entangled with other qubits.
The old qubit $a$ is selected to $+$; subsequent gates on that
logical wire act on the fresh qubit $b$, as shown in
Fig.~\ref{ghz:fig:hadamard}.

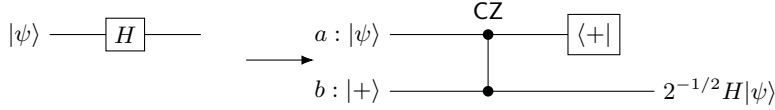
\begin{figure*}[ht]
\centering
\begin{tikzpicture}[x=1cm,y=0.75cm,font=\small]
 \draw (0,0)--(2,0);
 \node[left] at (0,0) {$\ket\psi$};
 \node[draw,fill=white] at (1,0) {$H$};
 \draw[-{Latex}] (2.6,-0.45)--(3.5,-0.45);
 \draw (4.5,0)--(7.2,0); \draw (4.5,-1)--(8,-1);
 \node[left] at (4.5,0) {$a:\ket\psi$};
 \node[left] at (4.5,-1) {$b:\ket+$};
 \draw (5.8,0)--(5.8,-1);
 \fill (5.8,0) circle (2pt); \fill (5.8,-1) circle (2pt);
 \node[above] at (5.8,0.12) {$\CZ$};
 \node[draw,fill=white] at (7.2,0) {$\bra+$};
 \node[right] at (8,-1) {$2^{-1/2}H\ket\psi$};
\end{tikzpicture}
\caption{Each replacement makes a fresh IQP qubit represent the
next segment of a logical wire. Time is recorded in the choice of
qubits and their interactions.}\label{ghz:fig:hadamard}
\end{figure*}

After replacing the Hadamards, all gates are diagonal. No later gate
uses a terminated qubit, so its measurement can be deferred to the
end. Prepare all fresh $\ket+$ qubits at the start and expand
\begin{align}\label{ghz:eq:cz}
 \CZ_{ab}=e^{-i\pi/4}R_a(\pi/4)R_b(\pi/4)G_{ab}(-\pi/4)\ .
\end{align}
Ignoring an overall phase, the resulting IQP unitary on $n$ qubits is
\begin{align}\label{ghz:eq:target}
 U=e^{i(F+B)},\quad
 F=\sum_v\alpha_v Z_v,\quad
 B=\sum_{e=\{u,v\}}\beta_e Z_uZ_v\ ,
\end{align}
where $\beta_e=-\pi/4$. Combine all rotations at a vertex into
one angle $\alpha_v\in(\pi/8)\mathbb Z$, taken modulo $2\pi$.
Each Hadamard replacement
contributes the same nonzero scalar $2^{-1/2}$ for every output.
Thus, conditioned on its selected outcomes, this IQP circuit has
exactly the output distribution of $C$.

Place the qubits for each logical wire on one horizontal row, in
their original time order. The Hadamard gadgets form horizontal
chains; the original CZ gates form vertical bonds between adjacent
rows. Each vertex has at most two horizontal bonds and one vertical
bond. This data graph is planar. Its vertices remain distinct after
the gates commute: the original time coordinate has become a second
spatial coordinate.

\subsection{Supply the \texorpdfstring{$Z$}{Z}-rotations with a GHZ reference}
\label{ghz:sec:ghz}
\subsubsection{The basic equality check}
Prepare an auxiliary qubit $a$ in $\ket+$, couple it to $u,v$ with
opposite angles, and select its $X$ outcome to zero. Its Kraus operator is
\begin{align}\label{ghz:eq:parity}
 \begin{split}
 {}_a\bra+G_{ua}(\pi/4)G_{av}(-\pi/4)\ket+_a
 &=\cos\!\left[\frac\pi4(Z_u-Z_v)\right]\\
 &=\frac{I+Z_uZ_v}{2}=:P^+_{uv}\ .
 \end{split}
\end{align}
It retains exactly equal computational-basis bits at $u,v$, without
a phase correction. Apply this check between consecutive reference
qubits $r_1,\ldots,r_m$, using $m-1$ auxiliaries. Then
\[
 P_R=\prod_{j=1}^{m-1}P^+_{r_jr_{j+1}}
 =\ket{0^m}\bra{0^m}+\ket{1^m}\bra{1^m}\ ,
\]
and
\begin{align}\label{ghz:eq:ghz-prep}
 \begin{aligned}
 P_R\ket+^{\otimes m}&=2^{-(m-1)/2}\GHZ{m}\ ,\\
 \GHZ{m}&=\frac{\ket{0^m}+\ket{1^m}}{\sqrt2}\ .
 \end{aligned}
\end{align}
The preparation succeeds with probability $2^{-(m-1)}$.
For $m=1$ there are no checks and $\GHZ1=\ket+$.

\begin{figure*}[ht]
\centering
\begin{tikzpicture}[x=1.55cm,y=1cm,font=\small]
 \foreach \j in {1,2,3,4} {
   \node[site] (r\j) at ({2*(\j-1)},0) {$r_\j$};
   \node[site] (v\j) at ({2*(\j-1)},-1.25) {$v_\j$};
   \draw (r\j)--node[right] {$\alpha_\j$} (v\j);
 }
 \foreach \j in {1,2,3} {
   \pgfmathtruncatemacro{\next}{\j+1}
   \node[check] (a\j) at ({2*\j-1},0) {$a_\j$};
   \draw (r\j)--node[above] {$+\pi/4$} (a\j);
   \draw (a\j)--node[above] {$-\pi/4$} (r\next);
 }
 \node at (3,0.8) {reference chain: all $r_j$ carry the same bit};
 \node at (3,-1.95) {select every $a_j$ to $+$; take the data qubits as the output register};
\end{tikzpicture}
\caption{One GHZ port for each required $Z$-rotation. Lines denote
$ZZ$ gates with the displayed angles. This is an interaction-graph
schematic; Section~\ref{ghz:sec:lattice} supplies a local king-lattice
placement. Extra reference qubits may lie between the ports.}\label{ghz:fig:ghz}
\end{figure*}
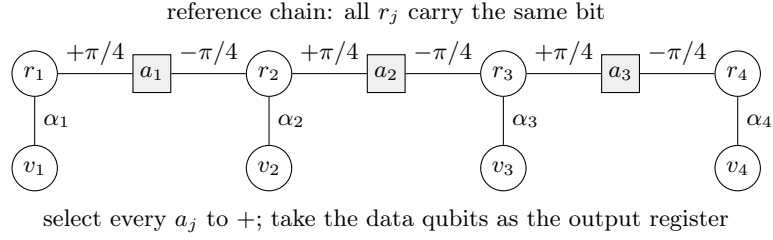

\subsubsection{Replacing the one-body rotations}
Choose a reference chain long enough to assign a distinct port
$r_{\iota(v)}$ to every data vertex $v$ (Fig.~\ref{ghz:fig:ghz}). Replace $R_v(\alpha_v)$
by $G_{v r_{\iota(v)}}(\alpha_v)$; a zero angle needs no coupling.
The other reference qubits simply extend the chain between ports.
The interaction is now
\begin{align}\label{ghz:eq:W}
 W=\exp\!\left(iB+i\sum_v
          \alpha_v Z_vZ_{r_{\iota(v)}}\right)\ .
\end{align}
On $\ket{0^m}_R$ it acts as $U_+=e^{i(B+F)}=U$; on $\ket{1^m}_R$
it acts as $U_-=e^{i(B-F)}$. The state after successful GHZ
preparation and these interactions is therefore
\[
 \frac{\ket{0^m}_R U_+\ket+^{\otimes n}
       +\ket{1^m}_R U_-\ket+^{\otimes n}}{\sqrt2}\ .
\]
Take the data qubits as the output register. Its distribution is
the marginal over the reference measurement outcomes, hence the
equal mixture of the two branch distributions.
These have identical $X$-measurement probabilities: with
$X_D=X^{\otimes n}$ on the data,
\begin{align}\label{ghz:eq:symmetry}
 \begin{gathered}
  U_-=X_DU_+X_D\ ,\\
  X_D\ket+^{\otimes n}=\ket+^{\otimes n}\ ,\\
  \bra{x}_XX_D=(-1)^{|x|}\bra{x}_X\ ,
 \end{gathered}
\end{align}
where $|x|=\sum_jx_j$. The last sign comes from
$X\ket+=\ket+$ and $X\ket-=-\ket-$, and disappears upon taking
the squared modulus. Writing $A_\pm(x)=\bra{x}_XU_\pm\ket+^{\otimes n}$,
we obtain, with $t$ denoting the reference outcomes and conditioned
on the GHZ-preparation checks,
\begin{align}\label{ghz:eq:reference-probability}
 \begin{aligned}
  p_D(x)&=\sum_t p_{DR}(x,t)\\
  &=\tfrac12|A_+(x)|^2+\tfrac12|A_-(x)|^2\\
  &=|A_+(x)|^2=p_U(x)\ .
 \end{aligned}
\end{align}
This equality holds jointly for all data outcomes, so it also
preserves any subsequent postselection on them.

\subsection{Lay the GHZ reference alongside the brickwork}
\label{ghz:sec:lattice}

\subsubsection{Data rows and a reference snake}
Keep the planar data layout just described. Put a reference row
below each data row and connect successive reference rows alternately
at the right and left boundary, outside the data region. This gives
one serpentine reference chain. Short connections from the data
vertices to their reference row supply the $Z$-rotations.
The horizontal data chains do not cross anything. A vertical CZ bond
between data rows $j$ and $j+1$ crosses only reference row $j$.
These are the only crossings to resolve. Figure~\ref{ghz:fig:brickwork}
shows the resulting layout.

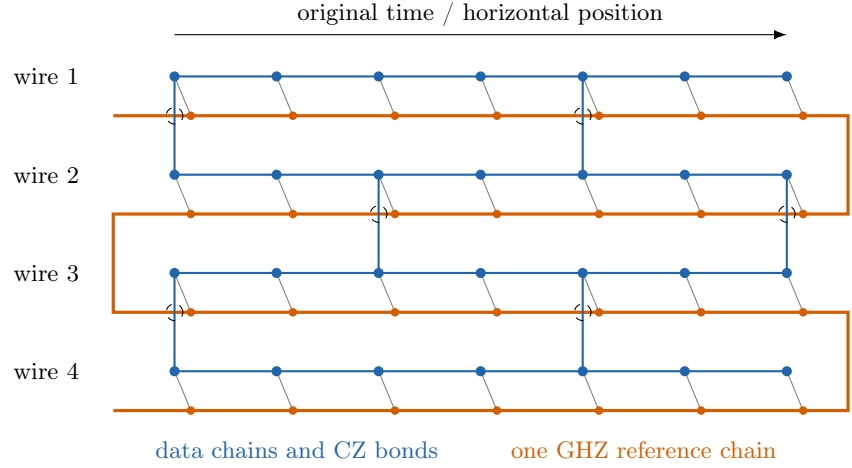
\begin{figure*}[ht]
\centering
\begin{tikzpicture}[x=1.35cm,y=0.65cm,font=\small]
 \definecolor{data}{RGB}{33,102,172}
 \definecolor{reference}{RGB}{213,94,0}
 \draw[reference,very thick]
 (-0.6,-0.8)--(6.6,-0.8)--(6.6,-2.8)--(-0.6,-2.8)
 --(-0.6,-4.8)--(6.6,-4.8)--(6.6,-6.8)--(-0.6,-6.8);
 \foreach \j in {0,1,2,3} {
   \pgfmathtruncatemacro{\wire}{\j+1}
   \draw[data,thick] (0,{-2*\j})--(6,{-2*\j});
   \node[left] at (-0.85,{-2*\j}) {wire $\wire$};
   \foreach \x in {0,1,2,3,4,5,6} {
     \draw[gray] (\x,{-2*\j})--({\x+0.16},{-2*\j-0.8});
     \fill[data] (\x,{-2*\j}) circle(1.9pt);
     \fill[reference] ({\x+0.16},{-2*\j-0.8}) circle(1.6pt);
   }
 }
 \foreach \x/\j in {0/0,0/2,2/1,4/0,4/2,6/1} {
   \draw[data,thick] (\x,{-2*\j})--(\x,{-2*\j-2});
   \draw[dashed] (\x,{-2*\j-0.8}) circle(0.11cm);
 }
 \draw[-{Latex}] (0,0.85)--(6,0.85);
 \node[above] at (3,0.85) {original time / horizontal position};
 \node[data] at (1.2,-7.6) {data chains and CZ bonds};
 \node[reference] at (4.6,-7.6) {one GHZ reference chain};
\end{tikzpicture}
\caption{Schematic of the construction. The brickwork data graph is planar. A reference snake visits
all data rows; grey connections supply the one-body rotations.
Circles mark the crossings of vertical CZ bonds with reference rows.
The fixed tile in Fig.~\ref{ghz:fig:crossings} resolves each crossing.
Intermediate qubits and parity-check auxiliaries are suppressed;
Fig.~\ref{ghz:fig:king-subgraph} shows a fully expanded king-lattice subgraph.}
\label{ghz:fig:brickwork}
\end{figure*}

Here is an explicit placement with room for the local gadgets.
Number the wires $j=0,\ldots,q-1$ and the circuit layers
$t=0,\ldots,T-1$. Put the layer boundaries at $(60t,60j)$.
In each horizontal interval of length $60$, the two-Hadamard case
has one intermediate data vertex at offset $30$; the three-Hadamard
case has two at offsets $20$ and $40$. The CZ bonds are vertical
segments at layer boundaries. These segments are disjoint within
each brickwork layer.

Put reference row $j$ at height $60j+20$ and extend it from
$x=-20$ to $x=60T+20$. Join the rows at alternating ends to form
the snake. A data vertex $v=(x,y)$ has its reference port at
$(x+6,y+20)$. Connect them by six diagonal steps
$(x+i,y+i)$, $i=0,\ldots,6$, followed by vertical steps to the port.
This small rightward offset separates the field connection from a
possible vertical CZ bond at column $x$.

\subsubsection{One crossing tile}
At a crossing, translate coordinates so that the surrounding square
is $\{0,\ldots,4\}^2$. Replace the horizontal reference segment and
the vertical data segment by the following two paths:
\begin{align*}
 \text{reference: }&(0,2),(1,2),(2,1),(3,2),(4,2)\ ,\\
 \text{data: }&(2,0),(1,1),(2,2),(2,3),(2,4)\ .
\end{align*}
They keep their endpoints and lengths, and share no lattice site.
The two diagonal edges cross between sites, so there is no qubit
or interaction at their geometric intersection. The spacing above
keeps the crossing tiles disjoint from each other, the data
vertices and the field connections.

\begin{figure*}[ht]
\centering
\begin{minipage}[t]{0.52\textwidth}
\centering
\subfloat[{A vertical data bond passes through a reference row without sharing a site.}]{\begin{minipage}{\linewidth}
\centering
\begin{tikzpicture}[x=0.7cm,y=0.7cm,font=\small]
 \draw[gray!30,step=1] (0,0) grid (4,4);
 \draw[orange!90!black,very thick] (0,2)--(1,2)--(2,1)--(3,2)--(4,2);
 \draw[blue!75!black,very thick] (2,0)--(1,1)--(2,2)--(2,3)--(2,4);
 \foreach \p in {(0,2),(1,2),(2,1),(3,2),(4,2)}
   {\fill[orange!90!black] \p circle(2.2pt);}
 \foreach \p in {(2,0),(1,1),(2,2),(2,3),(2,4)}
   {\fill[blue!75!black] \p circle(2.2pt);}
 \node[orange!90!black,left] at (0,2) {reference};
 \node[blue!75!black,above] at (2,4) {data bond};
\end{tikzpicture}
\end{minipage}}
\end{minipage}\hfill
\begin{minipage}[t]{0.38\textwidth}
\centering
\subfloat[{A detour adds one edge while preserving the path's endpoints.}]{\begin{minipage}{\linewidth}
\centering
\begin{tikzpicture}[x=0.7cm,y=0.7cm]
 \draw (0,3)--(0,1);
 \foreach \y in {1,2,3} {\fill (0,\y) circle(2pt);}
 \draw[-{Latex}] (1,2)--(2.3,2);
 \draw (3.3,3)--(4.3,2)--(3.3,2)--(3.3,1);
 \foreach \p in {(3.3,3),(4.3,2),(3.3,2),(3.3,1)} {\fill \p circle(2pt);}
\end{tikzpicture}
\end{minipage}}
\end{minipage}
\caption{Crossing and path-length adjustments in the king-lattice embedding.
Only the drawn edges are assigned nonzero couplings.}
\label{ghz:fig:crossings}
\end{figure*}

After these replacements, put reference qubits and check auxiliaries
alternately along the snake, starting with a reference qubit at its
left endpoint on the first row. Each prescribed port is a reference
qubit: all port positions have even distance along the original
snake, and the crossing tiles preserve these distances. Couple each
check auxiliary to its two reference neighbors with angles
$+\pi/4$ and $-\pi/4$, and select it to $+$.
Eq.~\eqref{ghz:eq:ghz-prep} then prepares one GHZ state along the
entire snake using local interactions.

\subsubsection{Realizing the data and field connections}
A drawn connection must implement its original $ZZ$ gate, including
when it passes through a crossing tile. The same parity check does
this on any odd-length path.
\begin{appendixlemma}
  \label{ghz:lem:transport}
On consecutive sites $u,a_1,c_1,\ldots,a_k,c_k,v$, put $c_0=u$ and set
\begin{align}\label{ghz:eq:path-gates}
 T_{uv}(\theta)=\left[\prod_{j=1}^k
 G_{c_{j-1}a_j}(\pi/4)G_{a_jc_j}(-\pi/4)\right]G_{c_kv}(\theta)\ .
\end{align}
Prepare and select all internal qubits in $\ket+$. Then
\begin{align}\label{ghz:eq:transport}
 {}_{a,c}\bra+^{\otimes 2k}T_{uv}(\theta)
 \ket+^{\otimes 2k}_{a,c}=2^{-k}G_{uv}(\theta)\ .
\end{align}
\end{appendixlemma}
\emph{Proof.}
Selecting the $a_j$ imposes $P^+_{c_{j-1}c_j}$. Every $c_j$ then has
the same computational-basis bit as $u$, so the last gate has the
phase of $G_{uv}(\theta)$. Each selected copy $c_j$ contributes
$\langle+|b\rangle\langle b|+\rangle=1/2$. The scalar is independent
of the endpoint state, so the identity also holds for entangled inputs.
\hfill$\square$

All data and field paths in our placement have even length before
adjustment. Add one unused site by replacing a unit edge with a
two-edge detour, as in Fig.~\ref{ghz:fig:crossings}. Explicitly, for a
horizontal data path beginning at $(x,y)$ use
\[
 (x+8,y)\longrightarrow(x+9,y-1)\longrightarrow(x+9,y)\ .
\]
For a vertical CZ path beginning at its lower endpoint $(x,y)$ use
\[
 (x,y+10)\longrightarrow(x-1,y+11)\longrightarrow(x,y+11)\ .
\]
For the field path from $(x,y)$ use
\[
 (x+6,y+10)\longrightarrow(x+7,y+11)\longrightarrow(x+6,y+11)\ .
\]
These detours avoid all other paths and crossing tiles. Apply
Lemma~\ref{ghz:lem:transport} to each resulting odd path, with
$\theta=-\pi/4$ for a data bond and $\theta=\alpha_v$ for a field
connection. Orient each field path from the reference port towards
the data vertex: its equality checks bring the reference bit next
to the data, where the final edge supplies the desired rotation.

Each data vertex has at most two horizontal connections, one
vertical CZ connection and one reference connection, hence degree
at most four. Reference ports have degree at most three, and all
other sites have degree at most two. Every gate angle is a multiple
of $\pi/8$, reduced modulo $2\pi$; in particular, a check angle
$-\pi/4$ equals $14\pi/8$ modulo $2\pi$. The rectangle has area
$O(q(T+1))$. Its remaining sites may contain uncoupled $\ket+$ qubits,
whose $X$ outcomes are deterministically zero.

\subsubsection{The resulting IQP circuit}
Every path gadget contributes a nonzero scalar independent of the
retained outcomes, while the GHZ reference preserves the data
output distribution. If $Q$ is the final circuit and $E$ is the event
that all prescribed gadget outcomes are zero, then
\begin{align}\label{ghz:eq:simulation}
 \Pr_Q(E)>0,\qquad
 \Pr_Q(z\mid E)=\left|\bra z_X C\ket+^{\otimes q}\right|^2\ .
\end{align}
Here $z$ lists the original circuit's output bits and the probability
on the left is their marginal distribution. Any postselections
belonging to $C$ can be imposed on the corresponding output bits.

All fresh qubits can be prepared at the start, and all gadget
postselections deferred to the end. Thus the king-lattice IQP circuit is
exactly Eq.~\eqref{ghz:eq:model}: prepare $\ket+$ at every site, apply
one $e^{i\mathcal H_J}$ containing only local two-body $ZZ$ terms,
and measure every qubit in $X$. The construction has polynomial
size, with no routing of a general interaction graph.

\begin{figure*}[p]
\centering
\begin{minipage}[t]{\textwidth}
\centering
\subfloat[{Complete expanded graph; the dashed box is enlarged in (b).}]{\begin{minipage}{\linewidth}
\centering
\includegraphics[width=\linewidth]{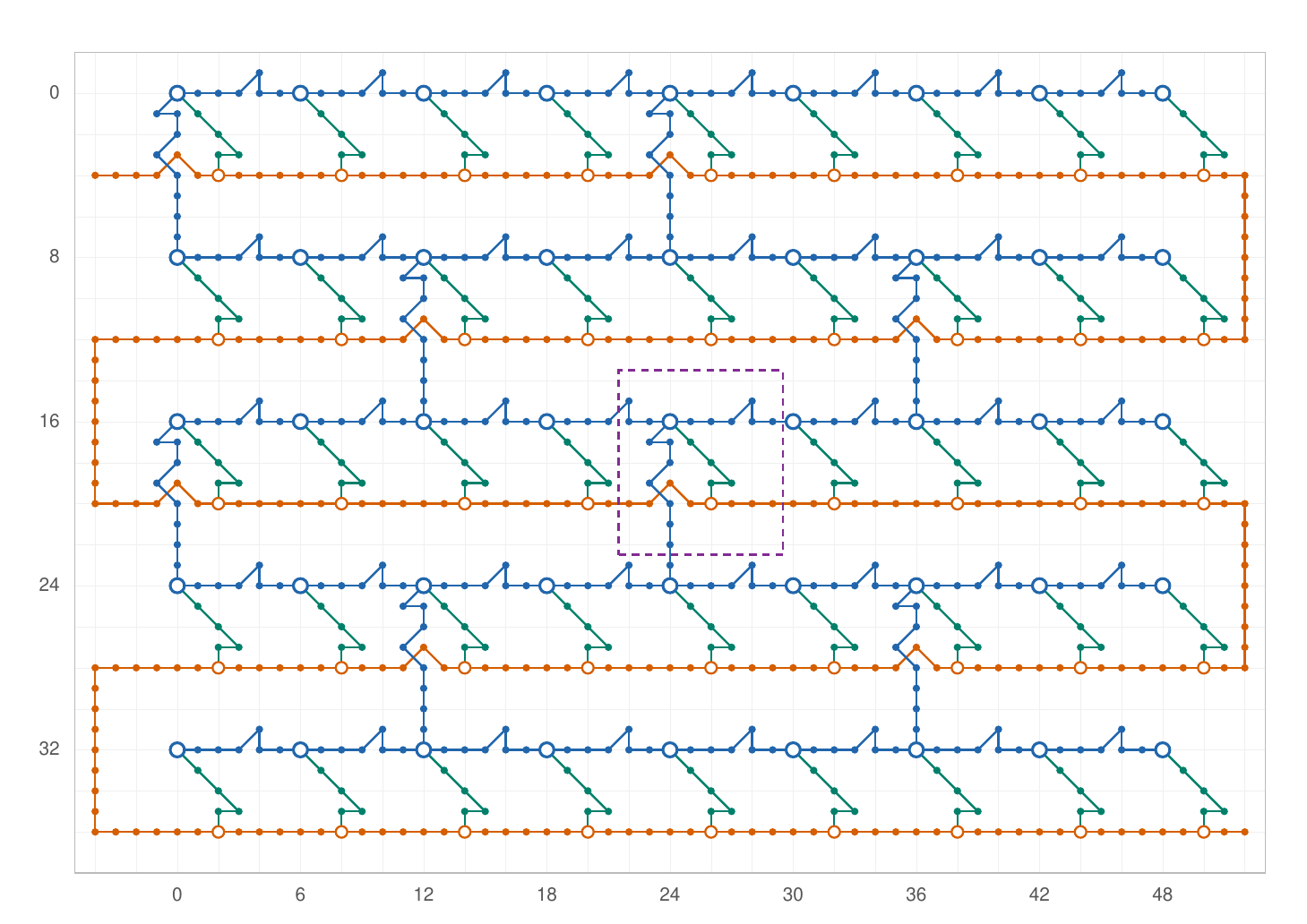}
\end{minipage}}
\end{minipage}
\par\medskip
\begin{minipage}[t]{0.85\textwidth}
\centering
\subfloat[{Bulk data site $(24,16)$ and reference port $(26,20)$. The data site has degree four; the crossing near $(24,20)$ has no shared vertex.}]{\begin{minipage}{\linewidth}
\centering
\includegraphics[height=5cm]{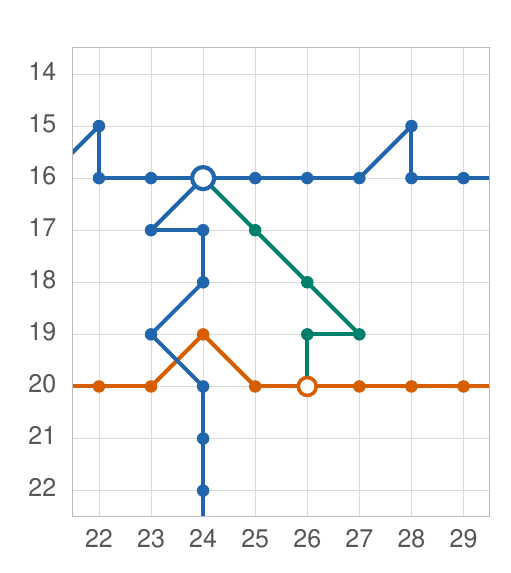}
\end{minipage}}
\end{minipage}
\caption{An actual king-lattice subgraph, including every auxiliary qubit, for five data rows and four staggered diagonal brickwork layers, each followed by the $HH$ padding. Every dot is a physical qubit; every colored segment is one king move (an edge of $\ell_\infty$ length one). Large blue rings mark original data vertices, and orange rings mark reference ports. Blue segments are data-interaction paths, orange segments form the GHZ reference/check chain, and green segments supply the $Z$-rotations. There are 842 active sites and 889 active couplings, with maximum degree four. Faint grid lines are coordinate guides; unused king-graph edges have zero coupling. This compact placement uses data-row spacing eight, horizontal data-vertex spacing six, and the same crossing and path-parity gadgets as the proof.}
\label{ghz:fig:king-subgraph}
\end{figure*}

\subsection{Physical embedding and multiplicative error}

In the geometric face coordinates of Sec.~\ref{sec:hardness} and the king coordinates of Eq.~\eqref{eq:kingcoords}, the rotated stabilizer domain is
\begin{align}
 \begin{split}
 \mathcal D_d=\{(a,b)\in\mathbb Z^2:{}&0\leq a+b\leq d-2\ ,\\
 &-1\leq a-b\leq d-1\}\ ,
 \end{split}
 \label{eq:kingdomain}
\end{align}
since $a+b$ is the geometric $x$ coordinate and $a-b$ the geometric $y$ coordinate. Every king edge with both endpoints in $\mathcal D_d$ is the syndrome support of a physical single-qubit or $ZZ$ term, as shown in Appendix~\ref{app:syndromeiqp}. As a check, we enumerated both sets for $d=5,7,9,11,13$; the numbers of king edges of $\mathcal D_d$ and of realized supports agree at $29$, $69$, $125$, $197$ and $285$.

The domain $\mathcal D_d$ contains arbitrarily large rectangular subsets as $d$ increases, so the finite rectangular layout constructed above can be translated into the bulk of a sufficiently large surface code. Its side lengths are polynomial in $q$ and $T$, and hence the required surface-code instance is polynomial in the size of the source circuit.

For each active king edge in this translated layout, choose one physical preimage and assign its single-qubit or $ZZ$-rotation angle the corresponding coupling $J_{uv}$; set all other physical angles to zero. The retained data outcomes, gadget postselections and GHZ-reference outcomes are then outcomes of physical syndrome-register qubits, while every unused register qubit is deterministically zero by Eq.~\eqref{eq:iqpinactivesites}. Marginalizing the GHZ-reference outcomes preserves the pointwise $\epsilon$-multiplicative bound, whereas conditioning on the gadget and source-circuit postselections changes conditional probabilities only by the factor $R_\epsilon=(1+\epsilon)/(1-\epsilon)$ used below.

\subsection{Multiplicative-error sampling consequence}\label{ghz:sec:complexity}
The starting circuits give postselected universal computation by
the construction of Ref.~\cite{BremnerJozsaShepherd2011}. Our exact polynomial-overhead
reduction preserves their conditional distributions. Conversely,
$\mathcal K$ uses ordinary quantum gates. Aaronson's postselection
theorem~\cite{Aaronson2005} therefore gives
\[
 \post\mathcal K=\post\BQP=\PP\ .
\]

For completeness, the postselection argument extends the exact-sampling consequence of the gate construction to every fixed $\epsilon<1$ in Eq.~\eqref{eq:mult}. Put $R_\epsilon=(1+\epsilon)/(1-\epsilon)$. Summing the pointwise bounds over events gives, for $p(P)>0$,
\begin{align}
 \widetilde p(E\mid P)\leq R_\epsilon p(E\mid P)\ ,
 \qquad\widetilde p(P)>0\ .
 \label{eq:conditionalerror}
\end{align}
Amplify a postselected quantum decision computation so that its conditional error is $\eta<1/(3R_\epsilon)$. Independent repetition, conditioning on all original selections, and a majority decision achieve this with a constant number of copies for fixed $\epsilon$. Compile the amplified circuit before applying the hypothetical sampler. On both yes and no instances, use Eq.~\eqref{eq:conditionalerror} for the event of an incorrect decision. The classical sampler conditioned on $P$ then has error below $1/3$. This is a postselected classical computation, not ordinary rejection sampling in polynomial time.

It follows that $\PP\subseteq\post\BPP$. Approximate counting gives $\post\BPP\subseteq\BPP^{\NP}$~\cite{Stockmeyer1983,AroraBarak2009}. Combining this containment with Toda's theorem and the relativized randomized-time containment yields~\cite{Toda1991,AroraBarak2009}
\begin{align}
 \begin{split}
 \PH&\subseteq\mathsf P^{\PP}
 \subseteq\mathsf P^{\post\BPP}
 \subseteq\mathsf P^{\BPP^{\NP}}\\
 &=\BPP^{\NP}\subseteq\Sigma_3^p\cap\Pi_3^p\ .
 \end{split}
\end{align}
Polynomially many oracle calls can be amplified to justify the equality. This proves Eq.~\eqref{eq:collapse} and hence Theorem~\ref{thm:hardness}. Exact sampling is included by setting $\epsilon=0$; constant total-variation error does not supply Eq.~\eqref{eq:conditionalerror} for rare postselection events.

\end
{document}